\documentclass[a4paper,fleqn]{cas-dc}

\usepackage[authoryear]{natbib}
\usepackage{amsmath,amssymb,amsfonts}
\usepackage{amsthm}

\newtheorem{assumption}{Assumption}
\newtheorem{proposition}{Proposition}
\newtheorem{theorem}{Theorem}

\newtheorem{remark}{Remark}

\usepackage{algorithm}
\usepackage{algpseudocode}

\usepackage{graphicx}
\usepackage{caption}
\usepackage{subcaption}
\usepackage{float}
\usepackage{stfloats}

\usepackage{booktabs}
\usepackage{multirow}
\usepackage{tabularx}
\usepackage{threeparttable}
\usepackage{makecell}

\usepackage{url}
\usepackage{pifont}
\usepackage{textcomp}

\def\tsc#1{\csdef{#1}{\textsc{\lowercase{#1}}\xspace}}
\tsc{WGM}
\tsc{QE}
\tsc{EP}
\tsc{PMS}
\tsc{BEC}
\tsc{DE}

\begin{document}
\let\WriteBookmarks\relax
\def\floatpagepagefraction{1}
\def\textpagefraction{.001}

\shorttitle{Cyber-Resilient Digital Twins}
\shortauthors{Homaei et~al.}

\title[mode=title]{Cyber-Resilient Digital Twins: Discriminating 
Attacks for Safe Critical Infrastructure Control}

\tnotetext[1]{This work is financed by the Junta de Extremadura 
and the European Union (ERDF funds) through the support funds to 
research groups GR24170.}

\author[1]{Mohammadhossein Homaei}[
    orcid=0000-0002-6108-6632]
\cormark[1]
\ead{mhomaein@alumnos.unex.es}
\credit{Conceptualization, Methodology, Software, 
Investigation, Writing -- Original Draft}

\author[2]{Iman Khazrak}[
    orcid=0000-0001-7087-2283]
\ead{ikhazra@bgsu.edu}
\credit{Methodology, Validation, Formal Analysis, Data Curation}

\author[1]{Rub\'{e}n Molano}[
    orcid=0000-0001-5410-6589]
\ead{rmolano@unex.es}
\credit{Formal Analysis, Methodology, Visualization}

\author[1]{Andr\'{e}s Caro}[
    orcid=0000-0002-6367-2694]
\ead{andresc@unex.es}
\credit{Supervision, Funding Acquisition, 
Writing -- Review \& Editing, Project Administration}

\author[1]{Mar \'{A}vila}[
    orcid=0000-0002-8717-442X]
\ead{mmavila@unex.es}
\credit{Supervision, Writing -- Review \& Editing}

\affiliation[1]{
    organization={Departamento de Ingenier\'{i}a de Sistemas 
        Inform\'{a}ticos y Telem\'{a}ticos, 
        Universidad de Extremadura},
    addressline={Av. Universidad S/N, Escuela Polit\'{e}cnica},
    city={C\'{a}ceres},
    postcode={10003},
    state={Extremadura},
    country={Spain}}

\affiliation[2]{
    organization={Department of Computer Science, 
        Bowling Green State University},
    city={Bowling Green},
    postcode={43403},
    state={OH},
    country={USA}}

\cortext[cor1]{Corresponding author}

\begin{abstract}
Industrial Cyber-Physical Systems (ICPS) face growing threats from cyber-attacks that exploit sensor and control vulnerabilities. Digital Twin (DT) technology can detect anomalies via predictive modelling, but current methods cannot distinguish attack types and often rely on costly full-system shutdowns. This paper presents i-SDT (intelligent Self-Defending DT), combining hydraulically-regularized predictive modelling, multi-class attack discrimination, and adaptive resilient control. Temporal Convolutional Networks (TCNs) with differentiable conservation constraints capture nominal dynamics and improve robustness to adversarial manipulations. A recurrent residual encoder with Maximum Mean Discrepancy (MMD) separates normal operation from single- and multi-stage attacks in latent space. When attacks are confirmed, Model Predictive Control (MPC) uses uncertainty-aware DT predictions to keep operations safe without shutdown. Evaluation on SWaT and WADI datasets shows major gains in detection accuracy, 44.1\% fewer false alarms, and 56.3\% lower operational costs in simulation-in-the-loop evaluation. with sub-second inference latency confirming real-time feasibility on plant-level workstations, i-SDT advances autonomous cyber-physical defense while maintaining operational resilience.
\end{abstract}

\vspace{1cm}
\begin{graphicalabstract}
\vspace*{3cm}
\makebox[\textwidth][c]{%
\includegraphics[width=0.8\textwidth]{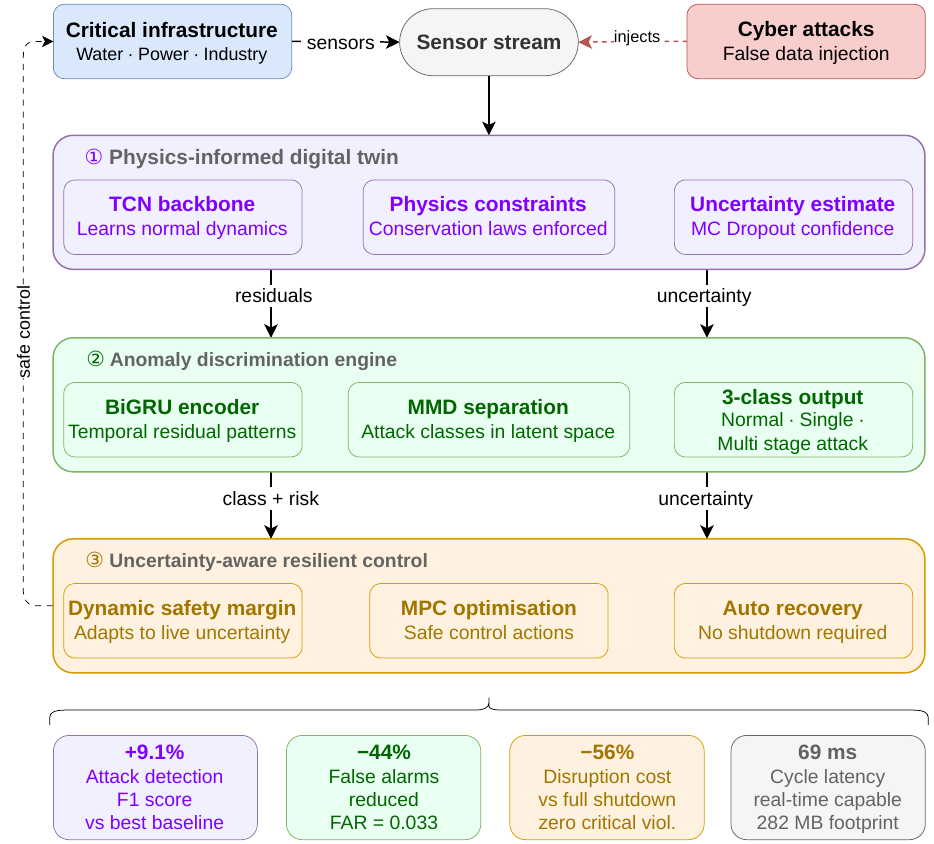}
}
\end{graphicalabstract}



\maketitle

\section{Introduction}
\label{sec:introduction}
Industrial Cyber-Physical Systems (ICPS) are the main foundation for critical infrastructures like water treatment plants, power grids, chemical factories, and transportation networks. In these systems, physical processes are connected to computing elements to monitor and control the operations in real time. Here, sensors $\mathcal{S}$, actuators $\mathcal{A}$, and control logic $\mathcal{C}$ work together to automate complex tasks. Although this connection makes systems much more efficient, it creates large surfaces for cyber-attacks~\citep{homaei2024review}. If attackers target sensors or controllers, they can change the physical processes, which leads to safety hazards, money loss, or a chain of failures. Because Industry~4.0 is bringing IT and operational networks together, systems that used to be isolated are now in danger from cybercriminals, hacktivists, and nation-state attackers~\citep{Lu2020}. We validate our i-SDT on water treatment testbeds in this work, but the proposed architecture can be applied to other domains simply by replacing~$\mathcal{L}_{\text{phys}}$.

Looking at past incidents shows how dangerous these risks are in reality. For example, in 2021, a hacker got into the Oldsmar water facility in Florida and almost increased the chemical levels to a dangerous point. This was only stopped because the operators noticed it in time. Before that, the Ukraine power grid attack in 2015 used weaknesses in the control systems and caused a blackout for 230,000 people. Also, the TRITON malware in 2017 attacked the safety instrumented systems (SIS) in a petrochemical plant. It tried to stop the emergency shutdown systems, which was the first time an attack specifically wanted to cause direct physical damage. Usually, these kinds of attacks use False Data Injection (FDI) methods. This helps them easily bypass normal IT Intrusion Detection Systems (IDS), because standard IDS cannot see small physical changes and cannot actively protect the system without the risk of stopping the production line~\citep{Lin2018}.

To solve this problem, Digital Twin (DT) technology can act as a very strong defense. It makes a virtual copy of the physical system that is always synced with the real sensor data. This allows the system to predict the normal behavior and easily find any strange changes or anomalies~\citep{Eckhart2019DTCSA, Lu2025DTMicrogrid, Repetto2026}.
 Modern DTs using LSTM, TCN, or autoencoders achieve over 85\% detection on SWaT and WADI \citep{Elnour2021LSTM}, \citep{Deng2021GDN}. However, they mostly classify events as normal or abnormal, cannot distinguish attacks from equipment faults, fail to separate single-stage $\mathcal{A}_S$ from multi-stage $\mathcal{A}_M$ attacks, and often rely on full shutdowns, causing high costs (\$100K–\$5M) and long restart times.

To address these challenges, we propose the i-SDT framework, featuring: (1) a physics-informed DT using mass-balance constraints for dynamic prediction (Section~\ref{sec:physics_scope}); (2) an Anomaly Discrimination Engine (ADE) classifying normal, $\mathcal{A}_S$, and $\mathcal{A}_M$ states via residual encoding and MMD; and (3) an uncertainty-aware resilient MPC ensuring safety during attacks. To prevent circular validation, we evaluate the controller using an independent hydraulic simulator (Section~\ref{sec:system_assumption}).
Evaluations on SWaT \citep{Mathur2016SWaT} and WADI \citep{Ahmed2017WADI} show i-SDT increases the $F1$-score by 9.1\% and cuts false alarms by 44.1\% over the best baseline. It also reduces operational costs by 56.3\%. i-SDT is highly efficient for real-time use, requiring only 69.0\,ms per cycle (well within $T_s{=}1$\,s) with a 282\,MB footprint. Our code is open-source.

The rest of this paper is organized as follows. Section~\ref{sec:related} reviews past work on ICPS anomaly detection, DT systems, and resilient control, showing where i-SDT stands. Section~\ref{sec:problem} defines the system model, main threats, and study goals. Section~\ref{sec:methodology} explains the i-SDT framework, including the PI DT, the ADE, and the resilient control part. Section~\ref{sec:results} reports the experiments on the SWaT and WADI datasets. Section~\ref{sec:discussion} discusses the results, deployment aspects, and limitations with future work. Finally, Section~\ref{sec:conclusion} summarizes the contributions and ends the paper.

\section{Related Work}\label{sec:related}

ICPS security has gained significant attention due to recent vulnerabilities in critical infrastructure. We review prior work in three areas: (1) anomaly detection, (2) DT applications, and (3) resilient control, leading to the identified gaps addressed by i-SDT.

\subsection{ICPS Anomaly Detection}
Initial ICPS intrusion detection adapted IT security methods to operational technology but failed against \textit{semantic attacks} where valid messages manipulate physical processes~\citep{Umer2022Survey}. Early data-driven approaches relied on basic autoencoders~\citep{Zhang2019} or LSTMs~\citep{Elnour2021LSTM}. Recently, the field shifted toward capturing complex multivariate dependencies using stochastic or adversarial frameworks like OmniAnomaly~\citep{Su2019OmniAnomaly} and USAD~\citep{Audibert2020USAD}. 

Currently, Graph Neural Networks (GNNs) and attention mechanisms, such as MTAD-GAT~\citep{Zhao2020MTADGAT} and Dual Attention Networks (DAN)~\citep{Xu2025DAN}, are standard for modeling spatial correlations among equipment components. Simultaneously, Transformer-based architectures like TranAD~\citep{Tuli2022TranAD} and LGAT~\citep{Wen2025LGAT} demonstrate superior capability in capturing long-term temporal dependencies. To handle stealthy Advanced Persistent Threats (APTs) at the edge, recent works focus on spatio-temporal dynamics and early-exit mechanisms to reduce computational costs~\citep{Luo2025STMBAD, Kong2025DADN}.

\subsection{Digital Twins for Cybersecurity}
Digital Twins have evolved from passive simulation tools into active cyber-defense mechanisms~\citep{Lu2020, Lu2025DTMicrogrid, Repetto2026}. Recent literature emphasizes integrating physical and process semantics to improve trustworthiness, utilizing process mining~\citep{Vitale2025ProcessMining} or hybrid physics-informed modeling~\citep{Hemmati2026HybridDT}. To overcome the limitations of pure correlation, Causal Digital Twins (CDT) were developed to enable robust anomaly detection and root-cause analysis~\citep{Homaei2026CausalDT}. Furthermore, the paradigm is shifting toward active mitigation. For instance, Zhang et al.~\citep{Zhang2025DTResilientMPC} recently proposed a DT-based resilient Model Predictive Control (MPC) framework that triggers a deterministic rollback mechanism to predict safe states during false data injection attacks.

\subsection{Resilient Control}
Traditional resilient control relies on accurate state estimation. Observer-based techniques aim to estimate actual states despite false data~\citep{Shafei2025SlidingModeObserver}, but standard active compensation often struggles against stealthy attacks due to reconstruction errors. Consequently, recent approaches focus on maintaining system trajectories within safety bounds without exact state reconstruction~\citep{Xin2025LearningResilient}. 

MPC is a prime candidate for resilient operation due to its constraint-handling capabilities~\citep{Rawlings2020MPC}. Recent advancements augment MPC against cyber-attacks using Stackelberg games~\citep{He2026} or switched cost functions against DoS attacks~\citep{Yang2025ResilientMPCDoS}. To counter stealthy FDI attacks, Xie et al.~\citep{Xie2025InescapableSet} designed a robust MPC (RMPC) integrated with inescapable sets to stabilize closed-loop dynamics.

\begin{table}[h]
\centering
\caption{Comparison of i-SDT with State-of-the-Art Approaches}
\label{tab:related_comparison}
\begin{threeparttable}
\resizebox{\columnwidth}{!}{%
\begin{tabular}{lcccc}
\toprule
\textbf{Approach} 
    & \makecell{\textbf{Physics} \\ \textbf{Informed}} 
    & \makecell{\textbf{Attack} \\ \textbf{Taxonomy}} 
    & \makecell{\textbf{Uncertainty} \\ \textbf{Aware}} 
    & \makecell{\textbf{Resilient} \\ \textbf{Control}} \\
\midrule
TranAD (2022)~\citep{Tuli2022TranAD} 
    & \ding{55} & \ding{55} & \ding{55} & \ding{55} \\
LGAT (2025)~\citep{Wen2025LGAT}      
    & \ding{55} & \ding{55} & \ding{55} & \ding{55} \\
CDT (2026)~\citep{Homaei2026CausalDT}
    & \ding{51} & \ding{55} & \ding{55} & \ding{55} \\
DT-MPC (2025)~\citep{Zhang2025DTResilientMPC}
    & \ding{55} & \ding{55} & \ding{55} & \ding{51} \\
Secure RMPC (2025)~\citep{Xie2025InescapableSet}
    & \ding{55} & \ding{55} & \ding{55} & \ding{51} \\
\midrule
\textbf{i-SDT (Proposed)}
    & \ding{51}$^{\dagger}$ & \ding{51} & \ding{51} & \ding{51} \\
\bottomrule
\end{tabular}%
}
\begin{tablenotes}
\footnotesize
\item[$\dagger$] Partial physics encoding: mass-balance constraint on 
tank-level \\ subsystem only. See Section~\ref{sec:physics_scope}.
\end{tablenotes}
\end{threeparttable}
\end{table}

\subsection{Identified Gaps and the i-SDT Positioning} \label{sec:gaps_positioning}

While the aforementioned approaches provide robust theoretical foundations, they exhibit critical limitations for dynamic ICPS. Table~\ref{tab:related_comparison} summarizes these gaps across four dimensions. i-SDT is designed to address the following issues within a single deployable framework:

\begin{enumerate}
    \item \textbf{Physics Blindness \& Binary Outputs:} Most data-driven detectors rely purely on statistical correlations and provide binary alerts. They fail to enforce physical laws and cannot distinguish between single-stage ($\mathcal{A}_S$) and multi-stage ($\mathcal{A}_M$) attacks.
    \item \textbf{Dependence on Known Attack Models:} Game-theoretic approaches~\citep{He2026} and conventional RMPC~\citep{Xie2025InescapableSet} rely on assumed upper bounds of attack vectors. Real-world stealthy threats frequently violate these predefined assumptions.
    \item \textbf{Lack of Data-Driven Uncertainty Integration:} Existing secure controllers and DT-MPC integrations~\citep{Zhang2025DTResilientMPC} separate detection from control logic. They do not dynamically adapt safety margins based on the real-time \textit{epistemic uncertainty} of the predictive model, leading to severe operational disruption costs ($\mathcal{D}_{\text{shutdown}}$) or unsafe control actions.
\end{enumerate}

To bridge these gaps, our proposed i-SDT framework unifies physics-informed prediction, MMD-based attack taxonomy discrimination, and probabilistic resilient control. Unlike deterministic approaches, i-SDT extracts real-time predictive uncertainty directly from the DT to create dynamic safety margins, ensuring safe operation without requiring a full system shutdown.

\section{Problem Statement and System Modeling}\label{sec:problem}

To enhance technical clarity and address the notation ambiguity, key symbols used throughout this manuscript are explicitly defined in Table~\ref{tab:nomenclature}.

\begin{table*}[t]
\scriptsize
\centering
\caption{Nomenclature and Variable Definitions}
\label{tab:nomenclature}
\begin{tabular}{@{}llll@{}}
\toprule
\textbf{Symbol} & \textbf{Description} 
    & \textbf{Symbol} & \textbf{Description} \\
\midrule

\multicolumn{4}{l}{\textit{System States and Signals}} \\[2pt]
$\mathbf{X}(t) \in \mathbb{R}^{d_x}$ 
    & Physical state vector at time $t$
    & $\mathbf{u}(t) \in \mathbb{R}^{d_u}$
    & Control input vector at time $t$ \\
$\mathbf{Y}_{\text{measured}}(t)$ 
    & Observed sensor measurements (potentially compromised)
    & $\hat{\mathbf{Y}}(t)$
    & Digital Twin predicted measurements \\
$\mathbf{R}(t)$ 
    & Residual vector: $\mathbf{Y}_{\text{measured}}(t) - \hat{\mathbf{Y}}(t)$
    & $\mathbf{R}_W(t) \in \mathbb{R}^{W \times d_y}$
    & Sliding residual window of length $W$ \\
$\mathbf{X}_\tau(t)$ 
    & Temporal context window $[\mathbf{X}(t{-}\tau),\ldots,\mathbf{X}(t{-}1)]$
    & $\mathbf{z}(t) \in \mathbb{R}^{2d_h}$
    & ADE latent embedding vector \\

\midrule
\multicolumn{4}{l}{\textit{Noise and Uncertainty}} \\[2pt]
$\boldsymbol{\omega}(t) \sim \mathcal{N}(\mathbf{0}, \Sigma_\omega)$
    & Process noise
    & $\boldsymbol{\epsilon}(t) \sim \mathcal{N}(\mathbf{0}, \Sigma_\epsilon)$
    & Measurement noise \\
$\Sigma_{\hat{Y}}(t)$
    & Predictive uncertainty covariance matrix from DT
    & $\bar{e}$
    & Empirical DT prediction error bound \\
$\sigma_{\text{data}}^2$
    & Data noise variance in DT loss
    & $\delta_{\text{MP}}(t)$
    & Model-plant discrepancy metric (Eq.~\ref{eq:mismatch}) \\

\midrule
\multicolumn{4}{l}{\textit{Attack and Threat Model}} \\[2pt]
$\boldsymbol{\delta}_{\mathcal{A}}(t) \in \mathbb{R}^{d_y}$
    & Injected FDI perturbation vector
    & $\beta_{\text{low}}$
    & Stealthiness bound on attack magnitude \\
$\mathcal{A}_S$
    & Single-stage attack class
    & $\mathcal{A}_M$
    & Multi-stage (coordinated) attack class \\
$\mathcal{N}$
    & Normal operation class
    & $\mathcal{C}(t)$
    & Predicted system status: $\{\mathcal{N}, \mathcal{A}_S, \mathcal{A}_M\}$ \\
$\mathcal{X}_{\text{unsafe}}$
    & Unsafe physical state region
    & $\epsilon_{\text{disc}}$
    & Attack discrimination gap (Eq.~\ref{eq:discrimination_gap}) \\

\midrule
\multicolumn{4}{l}{\textit{Physics-Informed Digital Twin}} \\[2pt]
$\boldsymbol{\Theta}_{\text{DT}}$
    & Trained DT (TCN) parameters
    & $\tau$
    & Temporal context window length \\
$\lambda_{\text{phys}}(t)$
    & Adaptive physics loss weight
    & $\alpha_0$
    & Base physics weight coefficient \\
$\theta_{\text{thresh}}$
    & Normalisation constant for $\lambda_{\text{phys}}$
    & $\Delta t$
    & Sampling interval (1\,s) \\
$\mathcal{I}_{\text{in}},\,\mathcal{I}_{\text{out}}$
    & Inflow / outflow sensor index sets
    & $V_{\text{tank}}$
    & Nominal tank volume (m$^3$) \\
$L,\, D,\, C$
    & Pipe length, diameter, and roughness coefficient
    & $\lambda_p$
    & Weight for pipe pressure loss penalty \\

\midrule
\multicolumn{4}{l}{\textit{Anomaly Discrimination Engine (ADE)}} \\[2pt]
$\boldsymbol{\Theta}_{\text{ADE}}$
    & Trained ADE (BiGRU) parameters
    & $W$
    & Residual window length (default 50) \\
$\beta$
    & MMD regularisation weight
    & $\gamma$
    & Classification confidence threshold \\
$\text{MMD}^2(\mathcal{Z}_1, \mathcal{Z}_2)$
    & Maximum Mean Discrepancy between distributions
    & $\zeta(t)$
    & Stealth risk metric from ADE probabilities \\
$\bar{\zeta}(t)$
    & EMA-smoothed risk metric
    & $\mu$
    & EMA filter coefficient (default 0.2) \\
$\omega$
    & Multi-stage attack weight in $\zeta$ (default 2.0)
    & $\varepsilon_0$
    & Numerical stability constant in $\zeta$ \\

\midrule
\multicolumn{4}{l}{\textit{Resilient MPC Controller}} \\[2pt]
$H$
    & MPC prediction horizon
    & $Q,\,R,\,P$
    & Stage, input, and terminal cost matrices \\
$A(k),\,B(k)$
    & Successive Jacobian linearisation matrices at step $k$
    & $K$
    & Offline LQR stabilising gain \\
$\mathbf{Y}_{\text{safe}}$
    & Safe operating setpoint
    & $\mathbf{Y}_{\min},\,\mathbf{Y}_{\max}$
    & Physical safety bounds \\
$\mathcal{Y}_f$
    & Terminal positively invariant set
    & $\boldsymbol{\epsilon}_{\text{margin}}(k)$
    & Robust safety margin for MPC constraints \\
$\kappa$
    & Uncertainty back-off coefficient (default 1.96)
    & $\eta$
    & Risk back-off coefficient (default 0.5) \\
$\bar{e}$
    & Lipschitz / mismatch bound (Assumption~\ref{ass:mismatch})
    & $L_f$
    & Lipschitz constant of $\mathcal{T}_{\text{DT}}$ \\
$\lambda_d$
    & Decay rate for stealth risk profile
    & $\zeta_{\min}$
    & Minimum baseline for decaying risk metric \\

\midrule
\multicolumn{4}{l}{\textit{Resilience and Cost Metrics}} \\[2pt]
$\mathcal{D}_{\text{shutdown}}$
    & Total disruption cost under full shutdown
    & $\mathcal{D}_{\text{rel}}$
    & Normalised disruption cost ($\mathcal{D}/\mathcal{D}_{\text{shutdown}}$) \\
$T_{\text{recovery}}$
    & Time to return to safe steady state (s)
    & $w_i,\,C_{\text{prod}},\,C_{\text{restart}}$
    & Importance weight, production loss rate, restart cost \\
$T_{\text{confirm}},\,T_{\text{val}}$
    & Recovery confirmation and validation windows
    & $\tau_{\text{recovery}}$
    & Model-plant discrepancy threshold for recovery \\
\bottomrule
\end{tabular}
\end{table*}

\subsection{Hybrid Evaluation Protocol: Datasets and Surrogate Environment}
\label{sec:system_assumption}

Evaluating simultaneous detection and mitigation in ICPS presents a methodological paradox: anomaly detection algorithms must be validated on real-world datasets containing authentic sensor noise and historical FDI patterns, whereas resilient control algorithms require a dynamic, closed-loop environment where state variables react to corrective actuations. Since the SWaT and WADI datasets are static, open-loop recordings, they cannot simulate post-actuation recovery.

To resolve this without compromising empirical validity, we adopt a decoupled, two-phase evaluation methodology. The Anomaly Discrimination Engine (ADE) is trained and evaluated strictly on the authentic SWaT/WADI datasets to ensure real-world detection fidelity. Conversely, to validate the resilient MPC component, we construct a two-tier surrogate evaluation environment. This environment deliberately separates the \textit{learned predictive model} from the \textit{evaluation plant}, providing a rigorous closed-loop testbed while avoiding circular validation.

\subsubsection{Tier~1 — Physics-Based Reference Simulator}
\label{sec:tier1}

A first-principles hydraulic simulator, denoted $\mathcal{P}_{\text{sim}}$, is constructed independently of the learned DT. For the SWaT testbed, $\mathcal{P}_{\text{sim}}$ implements the mass-balance dynamics of all six treatment stages using ordinary differential equations (ODEs) calibrated from the published SWaT physical specifications~\citep{Mathur2016SWaT}:

\begin{equation}\label{eq:ode_plant}
    \frac{d\mathbf{X}(t)}{dt} = \mathbf{f}_{\text{phys}}\!\left(\mathbf{X}(t),
    \mathbf{u}(t)\right) + \boldsymbol{\omega}(t),
    \quad \boldsymbol{\omega}(t)\sim\mathcal{N}(\mathbf{0},\Sigma_\omega),
\end{equation}

where $\mathbf{f}_{\text{phys}}$ encodes tank-level dynamics, pump flow rates, and inter-stage coupling independently derived from process engineering first principles, with no shared parameters with $\boldsymbol{\Theta}_{\text{DT}}$. For WADI, an analogous three-stage distribution network ODE is used. Nominal steady-state trajectories $\mathbf{X}^*$ are validated against the historical normal-operation segment of each dataset (NRMSE~$\leq 4.1\%$ for SWaT, $\leq 5.3\%$ for WADI).

\subsubsection{Tier~2 — Learned DT as Controller Oracle}
\label{sec:tier2}

The trained DT $\mathcal{T}_{\text{DT}}$ serves \emph{exclusively} as the prediction model inside the MPC optimisation (Eq.~\eqref{eq:method_mpc_objective}): it generates the forecast $\hat{\mathbf{Y}}(k{+}1)$ and the uncertainty estimate $\Sigma_{\hat{Y}}(k)$ over the horizon $[t,\, t{+}H{-}1]$. The control action $\mathbf{u}^*(t)$ computed by the MPC is then applied to $\mathcal{P}_{\text{sim}}$, \emph{not} back to the DT itself. The resulting closed-loop state trajectory from $\mathcal{P}_{\text{sim}}$ constitutes the ground truth against which safety constraint satisfaction, recovery time, and disruption cost are measured. 

This separation ensures that prediction errors of $\mathcal{T}_{\text{DT}}$ manifest as genuine tracking deviations in $\mathcal{P}_{\text{sim}}$, providing an honest assessment of control performance under model mismatch. The magnitude of this mismatch is quantified via the \textit{model-plant discrepancy metric}:

\begin{equation}\label{eq:mismatch}
    \delta_{\text{MP}}(t) \;=\;
    \bigl\|\mathbf{Y}_{\text{measured}}(t) - \hat{\mathbf{Y}}(t)\bigr\|_2
    \;\Big/\;
    \bigl\|\mathbf{Y}_{\text{measured}}(t)\bigr\|_2,
\end{equation}

whose empirical distribution over the test set is reported in
Section~\ref{sec:results_resilience} alongside the resilience metrics.

\subsubsection{Scope and Limitations: A Robust Control Perspective} 
\label{sec:scope}

We explicitly acknowledge a fundamental limitation: $\mathcal{P}_{\text{sim}}$ is an ODE-based mathematical representation and cannot perfectly replicate the complex, non-linear hydrodynamics of the physical SWaT and WADI hardware. Consequently, the reported resilient control metrics are strictly simulation-in-the-loop estimations.

However, from the perspective of Robust Control Theory, we formulate this limitation as an evaluation advantage. By artificially injecting increasing levels of model-plant mismatch ($\Sigma_{\omega}^{\text{perturbed}} = (1 + \gamma_{\text{mismatch}}) \Sigma_{\omega}^{\text{nominal}}$ with $\gamma_{\text{mismatch}} \in \{0.05, 0.10, 0.15, 0.20\}$), we construct a deliberate ``worst-case scenario'' environment. The objective of this tier is not to claim that our ODE is a flawless replica of the plant, but to demonstrate that the proposed controller maintains Input-to-State Stability (ISS) even when the predictive DT model significantly diverges from the evaluation plant. Demonstrating a 56.3\% disruption reduction under these degraded conditions proves that the uncertainty-aware margins ($\boldsymbol{\epsilon}_{\text{margin}}$) can safely absorb severe epistemic uncertainties without causing catastrophic failure.

\subsection{Industrial CPS Dynamics Model}\label{sec:system_model}
We formally model the ICPS progressing in discrete time $t \in \mathbb{N}$. The physical state evolution and sensor observation are governed by:
\begin{equation}\label{eq:normal_dynamics}
    \mathbf{X}(t+1) = f(\mathbf{X}(t), \mathbf{u}(t)) + \boldsymbol{\omega}(t),
\end{equation}

where $f(\cdot)$ models physical process and $\boldsymbol{\omega}(t) \sim \mathcal{N}(0, \Sigma_{\omega})$ represents bounded disturbances. Sensors report noisy measurements via
\begin{equation}\label{eq:normal_measurement}
    \mathbf{Y}_{\text{measured}}(t) = h(\mathbf{X}(t)) + \boldsymbol{\epsilon}(t),
\end{equation}
with observation function $h(\cdot)$ and noise $\boldsymbol{\epsilon}(t) \sim \mathcal{N}(0, \Sigma_{\epsilon})$ (Eq.~\eqref{eq:normal_measurement}).

This model reflects the standard abstraction of water and industrial processes; however, challenges arise when cyber manipulation affects data streams that are not directly observable from physical evolution alone.

\subsection{Adversarial Manipulation and Inconsistency}\label{sec:attack_model}
Cyber threat considered is false data injection (FDI) attack where adversary can perturb sensor channels or control paths. Under attack, measurement becomes
\begin{equation}\label{eq:fdi_attack}
    \mathbf{Y}_{\text{measured}}(t) = h(\mathbf{X}(t)) + \boldsymbol{\epsilon}(t) + \boldsymbol{\delta}_{\mathcal{A}}(t),
\end{equation}
where $\boldsymbol{\delta}_{\mathcal{A}}(t) \in \mathbb{R}^{d_y}$ denotes injected perturbation. Adversary aims to push physical state toward unsafe region $\mathcal{X}_{\text{unsafe}}$ while keeping observed data close to normal ranges, creating physical–cyber inconsistency: true dynamics follow (\ref{eq:normal_dynamics}), but reported data do not.

We consider single-stage attacks $\mathcal{A}_S$ and coordinated multi-stage attacks $\mathcal{A}_M$, where the latter manipulates several correlated variables over time. Stealthy attacks satisfy $\|\boldsymbol{\delta}_{\mathcal{A}}(t)\|  \leq \beta_{\text{low}}$, where $\beta_{\text{low}} > 0$ is a  dataset-specific bound below which naive residuals remain within  normal bounds, and are especially difficult for classical detectors  to identify.

\begin{quotation}
\noindent A theoretical ``perfect'' physics-aware attacker could bypass the DT by manipulating multiple sensors to satisfy all conservation equations simultaneously, forcing the physics residual $\mathcal{L}_{\text{phys}}$ to zero. However, executing such an attack necessitates simultaneous, white-box read/write access to all distributed Programmable Logic Controllers (PLCs) across isolated operational sub-networks. In real-world ICPS architectures, such omnipotent and synchronous access is operationally highly improbable. 

Therefore, our threat model specifically defines a \textit{Distributed Architecture Constrained} attacker. This adversary attempts to remain stealthy locally but lacks the global network access required to perfectly forge the system-wide mass-balance continuously. The key vulnerability exploited by our i-SDT framework is the \textbf{conservation-correlation trade-off}: because the attacker is constrained by the distributed architecture, their localized data injections inevitably break the broader multivariate statistical correlations. The Anomaly Discrimination Engine (ADE) captures this forced spatial-temporal distortion in the latent space, ensuring a robust defense-in-depth mechanism against realistic, topology-constrained adversaries.
\end{quotation}

\subsection{DT Prediction and Residual Behaviour}\label{sec:dt_model}
Data-driven DT, denoted $\mathcal{T}_{\text{DT}}$, is trained from normal operation. It learns approximate predictive mapping $\hat{f}$ from temporal context $\mathbf{X}_{\tau}(t) = [\mathbf{X}(t-\tau), \ldots, \mathbf{X}(t-1)]$. DT prediction is
\begin{equation}\label{eq:dt_prediction}
    \hat{\mathbf{Y}}(t) = \mathcal{T}_{\text{DT}}(\mathbf{X}_{\tau}(t-1); \boldsymbol{\Theta}_{\text{DT}}),
\end{equation}
where $\boldsymbol{\Theta}_{\text{DT}}$ are parameters.

Residual between measured and predicted value is
\begin{equation}\label{eq:residual}
    \mathbf{R}(t) = \mathbf{Y}_{\text{measured}}(t) - \hat{\mathbf{Y}}(t).
\end{equation}
Under normal conditions, residuals follow $\mathbf{R}(t) \sim \mathcal{N}(\mathbf{0}, \Sigma_R)$, 
where $\Sigma_R$ is the empirical residual covariance estimated from 
the normal-operation training set, due to sensor noise and modeling errors. During attacks, residuals develop temporal and cross-variable correlations reflecting manipulation nature. This motivates refined discrimination beyond classical anomaly detection.

\subsection{Identified Gaps in Existing Approaches}\label{sec:gaps}
Most anomaly detectors use simple binary view (Normal vs.\ Anomaly) with thresholds like $\|\mathbf{R}(t)\| > \tau$, ignoring difference between single-stage and multi-stage attacks and missing useful temporal patterns in coordinated actions.

Formally, ability to discriminate between $\mathcal{A}_S$ and $\mathcal{A}_M$ can be written as
\begin{equation}\label{eq:discrimination_gap}
    \epsilon_{\text{disc}}(\mathcal{A}_S, \mathcal{A}_M) =
    \left| \mathbb{E}_{\mathcal{A}_S}[\|\mathbf{R}(t)\|] -
    \mathbb{E}_{\mathcal{A}_M}[\|\mathbf{R}(t)\|] \right| \approx 0.
\end{equation}
Because both attack types may keep residual magnitude small, classical threshold-based detectors fail to separate them.

Second limitation appears in industrial mitigation strategies. Many facilities react to detected anomalies with full shutdown. Disruption cost over $[t_d, t_r]$ is
\begin{equation}\label{eq:disruption_cost}
\begin{aligned}
\mathcal{D}_{\text{shutdown}} =\;&
\sum_{i \in \mathcal{I}_{\text{crit}}} w_i 
\int_{t_d}^{t_r} \left| Y_i(t) - Y^*_{i,\text{safe}} \right| dt \\[4pt]
&+ C_{\text{prod}} (t_r - t_d)
+ C_{\text{restart}}.
\end{aligned}
\end{equation}

where $\mathcal{I}_{\text{crit}}$ denotes the set of critical sensor indices, $w_i > 0$ is a process-specific importance weight for  sensor $i$, $C_{\text{prod}}$ [\$/s] is the production loss rate, and $C_{\text{restart}}$ [\$] is the one-time restart cost.  The normalised metric $\mathcal{D}_{\text{rel}} =  \mathcal{D} / \mathcal{D}_{\text{shutdown}}$ is used for dataset-agnostic comparison in Section~\ref{sec:results_resilience}.

$\mathcal{D}_{\text{shutdown}}$ may reach very high values  (\$100K--\$5M per incident~\citep{homaei2024review}).  Existing mitigation methods do not exploit predictive ability of DTs, leading to unnecessary opportunity cost $\Delta\mathcal{D}$ (Eq.~\eqref{eq:opportunity_cost}).

\begin{equation}\label{eq:opportunity_cost}
    \Delta \mathcal{D} = \mathcal{D}_{\text{shutdown}} - \mathcal{D}_{\text{optimal}}.
\end{equation}

\subsection{Objectives and Scope of This Work}\label{sec:gaps_objectives}
This work addresses these gaps by designing DT-enhanced framework that (i) discriminates between normal behaviour, single-stage attacks, and multi-stage attacks, and (ii) performs mitigation without full shutdown relying on DT-based safe predictions. Classification task is
\begin{equation}\label{eq:classes}
    \mathcal{C}(t) \in \{\mathcal{N}, \mathcal{A}_S, \mathcal{A}_M\}.
\end{equation}
The goal is to achieve high discrimination accuracy (Eq.~\eqref{eq:classes}), reduced false alarms, and real-time performance for edge deployments. The mitigation strategy keeps the system inside safe operational boundaries while reducing the overall disruption cost relative to a conventional shutdown.

\section{Methodology}\label{sec:methodology}
We now describe the three tightly integrated modules of i-SDT. Each subsection covers the mathematical formulation, design rationale, and implementation details of one component.

\subsection{Framework Overview}\label{sec:method_overview}
i-SDT combines physics-informed DT (TCN with conservation constraints modeling nominal dynamics), ADE (bidirectional GRU with MMD regularization discriminating $\mathcal{N}$, $\mathcal{A}_S$, $\mathcal{A}_M$ from residuals), and resilient MPC (uncertainty-aware control during attacks). Offline training uses historical normal and simulated attack data; online deployment performs real-time monitoring and adaptive mitigation. Algorithm~\ref{alg:isdt_pipeline} presents the full pipeline from sensor measurements to prediction, classification, and control.

\begin{algorithm}[H]
\scriptsize  
\caption{i-SDT Online Pipeline}
\label{alg:isdt_pipeline}
\begin{algorithmic}[1]
\Require Trained models $\mathcal{T}_{\text{DT}}, \mathcal{F}_{\text{ADE}}$, threshold $\gamma$, window $W$
\Ensure Real-time classification $\mathcal{C}(t)$ and control $\mathbf{u}(t)$
\For{each timestep $t$}
    \State $\hat{\mathbf{Y}}(t) \gets \mathcal{T}_{\text{DT}}(\mathbf{X}_{\tau}(t))$; $\mathbf{R}(t) \gets \mathbf{Y}_{\text{measured}}(t) - \hat{\mathbf{Y}}(t)$
    \State Collect $\mathbf{R}_W(t)$; \textbf{if} $|\mathbf{R}_W(t)| = W$: $[\mathcal{C}(t), p] \gets \mathcal{F}_{\text{ADE}}(\mathbf{R}_W(t))$; \textbf{else}: $\mathcal{C}(t) \gets \mathcal{N}$, $p \gets 0$
\If{$\mathcal{C}(t) \in \{\mathcal{A}_S, \mathcal{A}_M\}$ \textbf{and} $p > \gamma$}
        \State $\mathbf{u}(t) \gets \text{ResilientMPC}(\hat{\mathbf{Y}}, \Sigma_{\hat{Y}}, \mathbf{Y}_{\text{safe}})$ \Comment{Alg.~\ref{alg:resilient_control}}
    \Else
        \State $\mathbf{u}(t) \gets \mathcal{C}_{\text{nom}}(\mathbf{Y}_{\text{measured}})$
    \EndIf
\EndFor
\end{algorithmic}
\end{algorithm}

\subsection{Physics-Guided DT}\label{sec:method_dt}

To visualize the information flow and variable dependencies within i-SDT, the comprehensive framework architecture is presented in Fig.~\ref{fig:architecture}.

\begin{center}
    \includegraphics[width=0.45\textwidth]{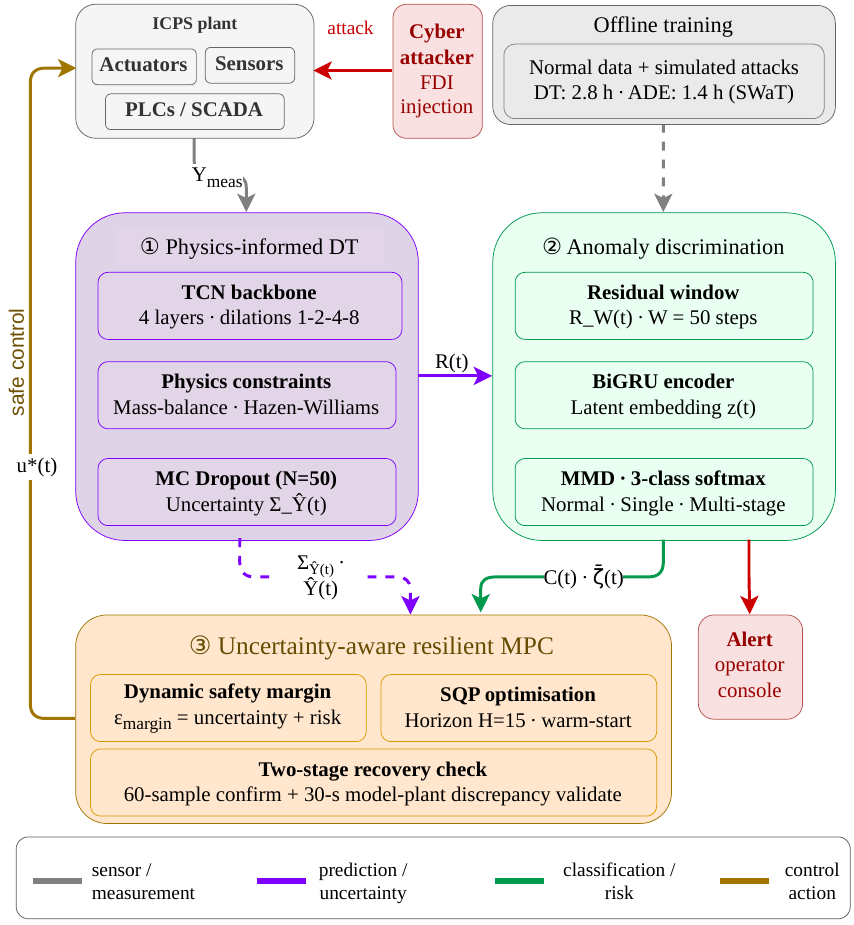}
    
    \captionof{figure}{The proposed i-SDT closed-loop architecture. (1) The PI-DT module leverages a TCN backbone constrained by mass-balance and hydraulic laws to predict normal states and quantify epistemic uncertainty via MC Dropout. (2) The ADE module isolates the spatial-temporal residual distortions using MMD to output risk metrics ($\zeta$). (3) The resilient MPC dynamically tightens safety margins using real-time uncertainty and risk profiles to compute safe operational commands ($u^*(t)$) without requiring system shutdown.}
    
    \label{fig:architecture}
\end{center}

The DT approximates nominal process dynamics by combining temporal deep learning with explicit physics constraints. Given a historical context window $\mathbf{X}_{\tau}(t) = [\mathbf{X}(t-\tau), \ldots, \mathbf{X}(t-1)] \in \mathbb{R}^{\tau \times d_x}$, the twin generates the future state prediction via a Temporal Convolutional Network (TCN):
\begin{equation}\label{eq:method_dt_forward}
    \hat{\mathbf{Y}}(t) = \mathcal{T}_{\text{DT}}(\mathbf{X}_{\tau}(t); \boldsymbol{\Theta}_{\text{DT}}) = \text{FC}\big(\text{TCN}(\mathbf{X}_{\tau}(t))\big),
\end{equation}
where the TCN backbone employs $L=4$ stacked residual blocks with dilation factors $\{1, 2, 4, 8\}$ to capture both short-term transients and long-range dependencies (see Fig.~\ref{fig:architecture}).

Training minimizes a composite objective that dynamically balances empirical accuracy with physical consistency. Unlike static weighting schemes, we employ an uncertainty-aware mechanism:

\begin{equation}\label{eq:method_dt_loss_bayesian}
\begin{aligned}
\mathcal{L}_{\text{DT}}(\boldsymbol{\Theta}_{\text{DT}})
&= \frac{1}{T}\sum_{t=1}^T
\Bigg\{
\frac{\|\mathbf{Y}_{\text{measured}}(t)-\hat{\mathbf{Y}}(t)\|_2^2}
{\sigma_{\text{data}}^2} \\
&\quad + \lambda_{\text{phys}}(t)\,
\mathcal{L}_{\text{phys}}(\hat{\mathbf{Y}})
\Bigg\}
\end{aligned}
\end{equation}

where the adaptive physics weight $\lambda_{\text{phys}}(t)$ is formulated to scale with the model's epistemic uncertainty:

\begin{equation}\label{eq:dynamic_weight}
\lambda_{\text{phys}}(t) = \alpha_0 \cdot \left( 1 + \tanh\left( 
\frac{\text{tr}(\Sigma_{\hat{Y}}(t))}{\theta_{\text{thresh}}} \right) \right).
\end{equation}

where $\theta_{\text{thresh}} > 0$ is a normalisation constant set to the median trace of $\Sigma_{\hat{Y}}$ over the training set, ensuring the $\tanh$ argument operates near its linear regime under nominal conditions.

 \paragraph{Rationale for Adaptive Weighting}
\label{sec:uncertainty_quantification}

Equation~\eqref{eq:dynamic_weight} addresses the ``trust'' dilemma in hybrid modelling. When the TCN encounters out-of-distribution samples (e.g., during attacks or rare operational states), the predictive uncertainty $\mathrm{tr}(\Sigma_{\hat{Y}}(t))$ increases, causing the $\tanh$ term to saturate towards~1 and increasing $\lambda_{\text{phys}}(t)$. This forces the optimisation to prioritise conservation laws (Eqs.~\ref{eq:physics_penalty_mass} and \ref{eq:physics_penalty_pressure}) over potentially corrupted sensor data, keeping the DT physically plausible when data reliability drops.

\paragraph{Scope of Physics Constraints}
\label{sec:physics_scope}

The physical regularisation term $\mathcal{L}_{\text{phys}}$ enforces a mass-balance and pressure constraint on the \textit{hydraulic subsystem} of the plant (involving flow, level, and pressure sensors). We intentionally restrict our framework to be \textit{hydraulically-constrained} rather than attempting to model the complete multi-physics or chemical state of the system. 

Specifically, our threat model strictly focuses on \textit{Cyber-Hydraulic Attacks}. Hydraulic variables exhibit fast dynamics with sub-minute time constants. Anomalies here cause immediate, catastrophic structural damage (e.g., pump dry-running or pipe rupture), explicitly requiring autonomous, sub-second MPC mitigation. Conversely, cyber-chemical attacks manipulating pH or chlorine involve slow-moving thermodynamic processes. Defending against such advanced persistent threats requires complex stoichiometric digital twins, falling outside our scope. This targeted hydraulic-constraint approach provides two benefits: it ensures critical structural stability against fast-acting attacks, and it minimizes the computational overhead for real-time plant-level execution. We note that actuator-level attacks such as direct PLC command injection fall outside the sensor-FDI threat model considered here; extending i-SDT to actuator compromise is identified as a priority for future work. Consequently, the differentiable constraint acts as a strategic regulariser, preventing the TCN from learning hydraulically impossible transitions without the prohibitive burden of full-plant simulation.

For systems where additional physical laws are available in differentiable form (e.g., Kirchhoff's laws for power grids, momentum balance for pipeline networks), the physics loss can be extended by adding corresponding terms to the total physics loss ($\mathcal{L}_{\text{phys}}$) without modifying the rest of the i-SDT architecture. 

Crucially, to preserve gradient flow during backpropagation, the temporal derivative is implemented as a differentiable linear operator:

\begin{equation}\label{eq:diff_operator}
\frac{d\hat{Y}_{\text{level}}(t)}{dt} \approx \frac{\hat{Y}_{\text{level}}(t) - \hat{Y}_{\text{level}}(t-1)}{\Delta t}.
\end{equation}

Since this operation is equivalent to a fixed-kernel convolution, it allows the gradients of the physics loss to be computed via automatic differentiation (Autograd) with respect to the TCN parameters $\boldsymbol{\Theta}_{\text{DT}}$. The mass balance residual is thus defined as:

\begin{equation}\label{eq:physics_penalty_mass}
\mathcal{L}_{\text{mass}}(\hat{\mathbf{Y}}) = \frac{1}{T}\sum_{t=1}^T 
\left\| \sum_{i \in \mathcal{I}_{\text{in}}} \hat{Y}_i(t) - 
\sum_{j \in \mathcal{I}_{\text{out}}} \hat{Y}_j(t) - 
V_{\text{tank}} \frac{d\hat{Y}_{\text{level}}}{dt} \right\|^2,
\end{equation}

To provide a rigorous physical grounding beyond simple volume conservation, we also incorporate the Hazen-Williams empirical equation for fluid dynamics, which governs the pressure drop $\Delta P$ across the primary piping networks:

\begin{equation}\label{eq:physics_penalty_pressure}
\mathcal{L}_{\text{pipe}}(\hat{\mathbf{Y}}) = \frac{1}{T}\sum_{t=1}^T 
\left\| \Delta \hat{P}(t) - 10.67 \cdot L \cdot \frac{\hat{Y}_{\text{flow}}(t)^{1.852}}{C^{1.852} \cdot D^{4.87}} \right\|^2,
\end{equation}

where $L$ (m) is the pipe length, $D$ (m) is the internal diameter, $\hat{Y}_{\text{flow}}$ is the volumetric flow rate ($m^3/s$), and $C$ is the dimensionless Hazen-Williams roughness coefficient specific to the SWaT testbed parameters \citep{Mathur2016SWaT}. The empirical constant 10.67 is rigorously calibrated for these SI units to ensure physical consistency in the loss function $\mathcal{L}_{\text{pipe}}$. The total physics regularization term is defined as the weighted sum $\mathcal{L}_{\text{phys}} = \mathcal{L}_{\text{mass}} + \lambda_p \mathcal{L}_{\text{pipe}}$. By backpropagating through these non-linear physical operators, the TCN explicitly learns the constrained manifold of valid hydraulic state transitions.

\subsection{Temporal Residual Encoding and Multi-Class Discrimination}\label{sec:method_residual}

ADE analyzes temporal patterns in residuals to differentiate normal variations from attacks. We compute $\mathbf{R}(t) = \mathbf{Y}_{\text{measured}}(t) - \hat{\mathbf{Y}}(t)$ and aggregate into windows:
\begin{equation}\label{eq:residual_window}
    \mathbf{R}_W(t) = [\mathbf{R}(t-W+1), \ldots, \mathbf{R}(t)] \in \mathbb{R}^{W \times d_y},
\end{equation}
with $W = 50$ capturing transient and drift patterns of stealthy attacks (Eq.~\ref{eq:residual_window}), consistent with the temporal residual analysis configured in our discrimination engine.

A bidirectional GRU network processes these temporal windows to extract discriminative features:
\begin{align}
    \mathbf{h}_{\text{fwd}}(t), \mathbf{h}_{\text{bwd}}(t) 
        &= \text{BiGRU}(\mathbf{R}_W(t); \boldsymbol{\Theta}_{\text{GRU}}),
        \label{eq:bigru_encoding} \\
    \mathbf{z}(t) 
        &= [\mathbf{h}_{\text{fwd}}(t); \mathbf{h}_{\text{bwd}}(t)] 
           \in \mathbb{R}^{2d_h},
        \label{eq:latent_concat}
\end{align}
Equations~\eqref{eq:bigru_encoding}--\eqref{eq:latent_concat} encode the residual window into a fixed-size latent vector $\mathbf{z}(t) \in \mathbb{R}^{2d_h}$ that captures asymmetric temporal patterns in both the causal and anti-causal directions of residual evolution.

The classification into our three-class taxonomy employs a softmax layer:
\begin{equation}\label{eq:method_softmax}
    P(\mathcal{C}(t) = c \mid \mathbf{z}(t)) = \frac{\exp(\mathbf{w}_c^\top \mathbf{z}(t) + b_c)}{\sum_{c' \in \{\mathcal{N}, \mathcal{A}_S, \mathcal{A}_M\}} \exp(\mathbf{w}_{c'}^\top \mathbf{z}(t) + b_{c'})}
\end{equation}

To enhance the model's ability to distinguish between single-stage and multi-stage attacks, we augment the standard cross-entropy loss with Maximum Mean Discrepancy regularization:
\begin{equation}\label{eq:method_discrim_loss}
\begin{aligned}
    \mathcal{L}_{\text{ADE}} 
    &= -\frac{1}{T}\sum_{t=1}^T 
       \log P\!\left(\mathcal{C}(t) = c_t \mid \mathbf{z}(t)\right) \\
    &\quad +\; \beta \cdot 
       \text{MMD}^2\!\left(\mathcal{Z}_{\mathcal{A}_S},\, \mathcal{Z}_{\mathcal{A}_M}\right).
\end{aligned}
\end{equation}

where the MMD term measures distributional divergence between attack classes using a Gaussian RBF kernel:
\begin{equation}\label{eq:mmd}
    \text{MMD}^2(\mathcal{Z}_1, \mathcal{Z}_2) = \mathbb{E}[k(\mathbf{z}_i, \mathbf{z}_j)] - 2\mathbb{E}[k(\mathbf{z}_i, \mathbf{z}_k)] + \mathbb{E}[k(\mathbf{z}_m, \mathbf{z}_n)]
\end{equation}

MMD advantages: (1) no hard negative mining unlike triplet loss; (2) captures higher-order statistics beyond covariance; (3) naturally handles class imbalance. Ablation with $\beta \in [0, 0.20]$ shows $\beta=0.1$ optimal. Our ablation with $\beta \in [0, 0.20]$ (Table~\ref{tab:ablation}) shows $\beta=0.1$ gives optimal separation.

This regularization improves contrastive learning, which clusters samples tightly and can hide class variation. MMD captures higher-order differences between attack distributions, keeping them distinct. Training process in Algorithm~\ref{alg:ade_training}.

\begin{algorithm}[H]
\scriptsize
\caption{ADE Training with MMD Regularization}
\label{alg:ade_training}
\begin{algorithmic}[1]
\Require Labeled residual dataset $\mathcal{D}_{\text{ADE}} = \{(\mathbf{R}_W^{(i)}, c^{(i)})\}_{i=1}^M$ with $c^{(i)} \in \{\mathcal{N}, \mathcal{A}_S, \mathcal{A}_M\}$
\Require Hyperparameters: learning rate $\eta$, MMD weight $\beta$, batch size $B$, epochs $E$, RBF kernel bandwidth $\sigma_k$ (default: median heuristic)
\Ensure Trained ADE parameters $\boldsymbol{\Theta}_{\text{ADE}}^*$
\State Initialize $\boldsymbol{\Theta}_{\text{ADE}} = \{\boldsymbol{\Theta}_{\text{GRU}}, \mathbf{W}_{\text{cls}}, \mathbf{b}_{\text{cls}}\}$ randomly
\For{epoch $e = 1$ to $E$}
    \State Shuffle $\mathcal{D}_{\text{ADE}}$
    \For{each mini-batch $\mathcal{B}$ of size $B$}
        \State Forward pass: $\mathbf{z}^{(i)} \gets \text{BiGRU}(\mathbf{R}_W^{(i)}; \boldsymbol{\Theta}_{\text{GRU}})$ for all $i \in \mathcal{B}$
        \State Compute class probabilities: $P(\mathcal{C}^{(i)} \mid \mathbf{z}^{(i)})$ using Eq.~\eqref{eq:method_softmax}
        \State Cross-entropy loss: $\mathcal{L}_{\text{CE}} \gets -\frac{1}{B}\sum_{i \in \mathcal{B}} \log P(\mathcal{C}^{(i)} = c^{(i)} \mid \mathbf{z}^{(i)})$
        \State Partition batch: $\mathcal{Z}_{\mathcal{A}_S} \gets \{\mathbf{z}^{(i)} : c^{(i)} = \mathcal{A}_S\}$, $\mathcal{Z}_{\mathcal{A}_M} \gets \{\mathbf{z}^{(i)} : c^{(i)} = \mathcal{A}_M\}$
        \If{$|\mathcal{Z}_{\mathcal{A}_S}| > 0$ \textbf{and} $|\mathcal{Z}_{\mathcal{A}_M}| > 0$}
            \State MMD regularization: $\mathcal{L}_{\text{MMD}} \gets \text{MMD}^2(\mathcal{Z}_{\mathcal{A}_S}, \mathcal{Z}_{\mathcal{A}_M})$ using Eq.~\eqref{eq:mmd}
        \Else
            \State $\mathcal{L}_{\text{MMD}} \gets 0$
        \EndIf
        \State Total loss: $\mathcal{L} \gets \mathcal{L}_{\text{CE}} + \beta \cdot \mathcal{L}_{\text{MMD}}$
        \State Update: $\boldsymbol{\Theta}_{\text{ADE}} \gets \boldsymbol{\Theta}_{\text{ADE}} - \eta \cdot \nabla_{\boldsymbol{\Theta}_{\text{ADE}}} \mathcal{L}$
    \EndFor
\EndFor
\Return $\boldsymbol{\Theta}_{\text{ADE}}^*$
\end{algorithmic}
\end{algorithm}

A key challenge is the limited availability of labelled attack data. To address this, we use simulation and transfer learning: the encoder is trained on simulated attacks, after which the normal-class boundary is fine-tuned using 7--14 days of real operational data with the attack layers frozen, achieving $>90\%$ of fully supervised performance. The classification procedure is detailed in Algorithm~\ref{alg:ade_classify}.

\subsection{Probabilistic Resilient Control with Formal Safety Guarantees}
\label{sec:method_control}

Upon confirming an attack via ADE (Algorithm~\ref{alg:ade_classify}), i-SDT transitions to resilient control mode. We first establish the theoretical foundations of the proposed MPC formulation before presenting the optimisation problem.

\subsubsection{Assumptions}
\label{sec:mpc_assumptions}

The following assumptions underpin the stability and feasibility analysis.

\begin{assumption}[Bounded Model-Plant Mismatch]\label{ass:mismatch}
The prediction error of the DT is bounded almost surely:

$\|\mathbf{Y}_{\text{measured}}(t) - \hat{\mathbf{Y}}(t)\|_2 \leq
\bar{e}$ for all $t$, where $\bar{e} > 0$ is estimated empirically
from the validation set.
\end{assumption}

\begin{assumption}[Lipschitz DT Dynamics]\label{ass:lipschitz}
The learned mapping $\mathcal{T}_{\text{DT}}$ is Lipschitz continuous
with constant $L_f$, i.e.,
$\|\mathcal{T}_{\text{DT}}(\mathbf{x}_1) -
 \mathcal{T}_{\text{DT}}(\mathbf{x}_2)\|_2 \leq
 L_f \|\mathbf{x}_1 - \mathbf{x}_2\|_2$.
\end{assumption}

\begin{assumption}[Compact Safe Set]\label{ass:safeset}
The safe operating region
$\mathcal{Y}_{\text{safe}} = \{\mathbf{Y} :
\mathbf{Y}_{\min} \leq \mathbf{Y} \leq \mathbf{Y}_{\max}\}$
is a compact convex polytope, and the tightened set
$\mathcal{Y}_{\text{tight}}(k) = \{\mathbf{Y} :
\mathbf{Y}_{\min} + \boldsymbol{\epsilon}_{\text{margin}}(k) \leq
\mathbf{Y} \leq
\mathbf{Y}_{\max} - \boldsymbol{\epsilon}_{\text{margin}}(k)\}$
is non-empty for all $k$.
\end{assumption}

\subsubsection{Successive Linearisation of the Learned DT}
\label{sec:linearisation}

Since $\mathcal{T}_{\text{DT}}$ is a highly non-linear network and industrial processes contain discrete operational shifts (e.g., pump switching), a single steady-state Jacobian is invalid under dynamic attack conditions. Instead, we employ Successive Linearisation (SL-MPC). At each control step $t$, the Jacobian matrices are dynamically evaluated around the current predicted state $\hat{\mathbf{Y}}(t)$ and the previous applied control input $\mathbf{u}(t-1)$:

\begin{equation}\label{eq:jacobian}
    A(t) = \frac{\partial\,\mathcal{T}_{\text{DT}}}{\partial\,\mathbf{Y}}\bigg|_{\hat{\mathbf{Y}}(t),\mathbf{u}(t-1)}, \quad B(t) = \frac{\partial\,\mathcal{T}_{\text{DT}}}{\partial\,\mathbf{u}}\bigg|_{\hat{\mathbf{Y}}(t),\mathbf{u}(t-1)}
\end{equation}

computed efficiently via automatic differentiation through the TCN computational graph. To strictly bind the Lipschitz constant $L_f$ (satisfying Assumption~\ref{ass:lipschitz}) and prevent unbounded Jacobian spectral norms during out-of-distribution adversarial transitions, we apply Spectral Normalization to the convolutional weights during training. The empirical Lipschitz estimate, computed via power iteration on the validation set, yields $L_f \leq 2.31$ for SWaT and $L_f \leq 2.74$ for WADI. Consequently, the linearisation residual reliably satisfies $\|r\|_2 \leq \bar{e}$ in the local neighborhood of the current trajectory, ensuring Assumption~\ref{ass:mismatch} holds during transient recovery phases.

\subsubsection{MPC Optimisation Problem}
\label{sec:mpc_opt}

The controller solves the following finite-horizon problem at each timestep $t$:

\begin{equation}\label{eq:method_mpc_objective}
\begin{aligned}
\min_{\{\mathbf{u}(k)\}_{k=t}^{t+H-1}} \quad
& \sum_{k=t}^{t+H-1}
\left(
  \|\hat{\mathbf{Y}}(k{+}1) - \mathbf{Y}_{\text{safe}}\|_Q^2
  + \|\Delta\mathbf{u}(k)\|_R^2
\right) \\
& + \|\hat{\mathbf{Y}}(t{+}H) - \mathbf{Y}_{\text{safe}}\|_P^2
\end{aligned}
\end{equation}

subject to:
\begin{align}
    \hat{\mathbf{Y}}(k{+}1)
        &= \mathcal{T}_{\text{DT}}
           \!\left(\mathbf{X}_\tau(k),\,\mathbf{u}(k)\right),
        \label{eq:mpc_dynamics} \\
    \mathbf{u}_{\min}
        &\leq \mathbf{u}(k) \leq \mathbf{u}_{\max},
        \label{eq:mpc_actuator} \\
    \mathbf{Y}_{\min} + \boldsymbol{\epsilon}_{\text{margin}}(k)
        &\leq \hat{\mathbf{Y}}(k) \leq
         \mathbf{Y}_{\max} - \boldsymbol{\epsilon}_{\text{margin}}(k),
        \label{eq:mpc_safety_robust} \\
    \hat{\mathbf{Y}}(t{+}H)
        &\in \mathcal{Y}_f.
        \label{eq:terminal_constraint}
\end{align}

The terminal set $\mathcal{Y}_f \subseteq \mathcal{Y}_{\text{safe}}$ in constraint~\eqref{eq:terminal_constraint} is the maximal positively invariant set of the nominal closed-loop system under the LQR gain $K$, computed offline via standard invariant-set algorithms~\citep{Rawlings2020MPC}. The terminal cost matrix $P \succ 0$ solves the discrete-time Lyapunov equation $A_K^\top P A_K - P = -Q - K^\top R K$, where $A_K = A - BK$ is the closed-loop matrix under the unconstrained LQR gain $K$. This choice guarantees that the terminal cost is a Lyapunov function for the nominal unconstrained system.

\subsubsection{Recursive Feasibility}
\label{sec:feasibility}

\begin{proposition}[Recursive Feasibility]\label{prop:feasibility}
Under Assumptions~\ref{ass:mismatch}--\ref{ass:safeset}, suppose the MPC problem~\eqref{eq:method_mpc_objective} is feasible at time $t$. If the safety margin satisfies

\begin{equation}\label{eq:margin_condition}
    \boldsymbol{\epsilon}_{\text{margin}}(k) \;\geq\;
    \kappa\,\mathrm{diag}\!\left(\Sigma_{\hat{Y}}(k)^{1/2}\right)
    + \bar{e}\,\mathbf{1},
    \quad \forall\, k \in [t,\, t{+}H{-}1],
\end{equation}

then the MPC problem remains feasible at time $t{+}1$.
\end{proposition}

\begin{proof}
Let $\mathbf{U}^*(t) = \{\mathbf{u}^*(t), \ldots, \mathbf{u}^*(t{+}H{-}1)\}$ be the optimal solution at time $t$. Construct the candidate solution at $t{+}1$ as the shifted sequence:
\begin{equation}\label{eq:candidate_solution}
    \tilde{\mathbf{U}}(t{+}1) =
    \{\mathbf{u}^*(t{+}1), \ldots, \mathbf{u}^*(t{+}H{-}1),
    \mathbf{u}_{\text{safe}}\},
\end{equation}
where $\mathbf{u}_{\text{safe}}$ in~\eqref{eq:candidate_solution} drives the predicted state toward $\mathbf{Y}_{\text{safe}}$ and satisfies $\mathbf{u}_{\min} \leq \mathbf{u}_{\text{safe}} \leq \mathbf{u}_{\max}$.

For any $k \in [t{+}1,\, t{+}H{-}1]$, the prediction error at $t{+}1$ is bounded by Assumption~\ref{ass:mismatch}:

\begin{equation}\label{eq:pred_error_bound}
    \|\mathbf{Y}_{\text{measured}}(k) - \hat{\mathbf{Y}}(k)\|_2
    \leq \bar{e}.
\end{equation}

Since $\boldsymbol{\epsilon}_{\text{margin}}(k) \geq \bar{e}\,\mathbf{1}$ by~\eqref{eq:margin_condition} and~\eqref{eq:pred_error_bound}, the tightened constraint in~\eqref{eq:mpc_safety_robust} absorbs this error, so $\hat{\mathbf{Y}}(k) \in \mathcal{Y}_{\text{tight}}(k)$ implies $\mathbf{Y}_{\text{measured}}(k) \in \mathcal{Y}_{\text{safe}}$. By Assumption~\ref{ass:safeset}, $\mathcal{Y}_{\text{tight}}(k)$ is non-empty, so $\tilde{\mathbf{U}}(t{+}1)$ is feasible. Therefore, the MPC problem is feasible at $t{+}1$. \qed
\end{proof}

\subsubsection{Input-to-State Stability Under Attack}
\label{sec:iss}

\begin{theorem}[ISS of the Closed-Loop System]\label{thm:iss}
Under Assumptions~\ref{ass:mismatch}--\ref{ass:safeset} and Proposition~\ref{prop:feasibility}, the closed-loop system under the i-SDT MPC is Input-to-State Stable with respect to the disturbance $\boldsymbol{\delta}_{\mathcal{A}}(t)$. Specifically, there exists a class-$\mathcal{KL}$ function $\beta$ and a class-$\mathcal{K}$ function $\gamma$ such that:

\begin{equation}\label{eq:iss}
    \|\mathbf{Y}(t) - \mathbf{Y}_{\text{safe}}\|_2 \;\leq\;
    \beta\!\left(\|\mathbf{Y}(0) - \mathbf{Y}_{\text{safe}}\|_2,\, t\right)
    + \gamma\!\left(\sup_{\tau \leq t}
      \|\boldsymbol{\delta}_{\mathcal{A}}(\tau)\|_2\right)
\end{equation}
\end{theorem}

\begin{proof}[Proof Sketch]
Define the Lyapunov candidate $V(t) = \|\hat{\mathbf{Y}}(t) - \mathbf{Y}_{\text{safe}}\|_P^2$. By the terminal cost construction (Section~\ref{sec:mpc_opt}), $V(t)$ is non-increasing along nominal (attack-free) trajectories. Under attack with $\|\boldsymbol{\delta}_{\mathcal{A}}(t)\|_2 \leq \beta_{\text{high}}$ (where $\beta_{\text{high}}$ is a finite  upper bound on the attack magnitude, consistent with the ISS  framework~\citep{Rawlings2020MPC}), the perturbation enters additively through the residual $\mathbf{R}(t)$. Using Assumption~\ref{ass:lipschitz} and the bound in Assumption~\ref{ass:mismatch}, the one-step Lyapunov decrease satisfies:

\begin{equation}\label{eq:lyapunov_decrease}
    V(t{+}1) - V(t) \;\leq\;
    -\alpha_3\!\left(\|\mathbf{Y}(t) - \mathbf{Y}_{\text{safe}}\|_2\right)
    + \sigma\!\left(\|\boldsymbol{\delta}_{\mathcal{A}}(t)\|_2\right)
\end{equation}

where $\alpha_3 \in \mathcal{K}_\infty$ and $\sigma \in \mathcal{K}$.
Standard ISS arguments~\citep{Rawlings2020MPC} then yield~\eqref{eq:iss}.
\qed
\end{proof}

\begin{remark}[Chattering Prevention]\label{rem:chattering}

The stealth risk metric $\zeta(t)$ in~\eqref{eq:stealth_risk_zeta} is computed from uncalibrated softmax outputs, which may oscillate rapidly near decision boundaries. To prevent constraint chattering in the MPC, $\zeta(t)$ is passed through an exponential moving average filter before being used in~\eqref{eq:uncertainty_margin_total}:

\begin{equation}\label{eq:ema}
    \bar{\zeta}(t) = (1-\mu)\,\bar{\zeta}(t{-}1) + \mu\,\zeta(t),
    \quad \mu = 0.2,
\end{equation}
ensuring that $\boldsymbol{\epsilon}_{\text{margin}}(k)$ varies smoothly and the SQP solver encounters a stable constraint landscape. The filter coefficient $\mu=0.2$ corresponds to a time constant of approximately $4$ seconds at $1$\,Hz sampling, which is short relative to the hydraulic time constants (minutes) but long enough to suppress softmax noise.
\end{remark}

\subsubsection{Horizon Propagation of Uncertainty and Risk}
\label{sec:horizon_propagation}

The safety margin $\boldsymbol{\epsilon}_{\text{margin}}(k)$ in constraint~\eqref{eq:mpc_safety_robust} depends on two quantities that must be evaluated at each future step $k \in [t,\,t{+}H{-}1]$: the predictive uncertainty $\Sigma_{\hat{Y}}(k)$ and the smoothed risk metric $\bar{\zeta}(t)$. We now specify how each is propagated over the horizon.

\paragraph{Uncertainty Propagation.}
At the current step $t$, the DT produces the epistemic uncertainty bound $\Sigma_{\text{epi}}(t)$ via MC Dropout. To avoid unphysical compounding of model ignorance with inherent process noise, we explicitly initialize the prediction covariance with the epistemic estimate. Over the prediction horizon, the total covariance is propagated using the successive linearisation matrices $A(k)$:

\begin{equation}\label{eq:sigma_propagation}
    \Sigma_{\hat{Y}}(k+1) = A(k)\Sigma_{\hat{Y}}(k)A(k)^\top + \Sigma_\omega, \quad \text{with} \quad \Sigma_{\hat{Y}}(t) = \Sigma_{\text{epi}}(t)
\end{equation}

for $k = t,\ldots,t+H-2$, where $\Sigma_\omega$ is the constant aleatoric process noise covariance. This non-stationary recursion provides a realistic uncertainty envelope over the horizon without catastrophic margin explosion.

\paragraph{Risk Metric Over the Horizon.}
The ADE classifier operates on the current residual window and produces $\zeta(t)$ at step $t$ only. Predicting the exact future evolution of $\zeta$ is impossible without knowing the attacker's strategy. However, adopting a worst-case constant assumption for the entire horizon leads to severe actuator saturation and violates operational continuity. Therefore, we introduce an \textit{exponentially decaying risk profile} tied to the hydraulic settling time ($T_c$) of the process:

\begin{equation}\label{eq:zeta_horizon}
\begin{aligned}
\bar{\zeta}(k) \;=\;
& \max\!\Big(
\bar{\zeta}(t)\,\exp\!\big(-\lambda_d (k - t)\big), \\
& \zeta_{\min}
\Big), \quad \forall\, k \in [t,\,t{+}H{-}1]
\end{aligned}
\end{equation}

where the decay rate $\lambda_d = 1/T_c$ allows the controller to gradually relax the excessive safety margins as it projects further into the future, and $\zeta_{\min} = 0.1 \bar{\zeta}(t)$ ensures a minimum baseline of defensive readiness. This formulation guarantees recursive feasibility while preventing the MPC optimizer from commanding aggressive, self-destructive corrections due to artificially prolonged worst-case assumptions.

\paragraph{Resulting Margin.}
Combining both propagations, the safety margin at step $k$ is:

\begin{equation}\label{eq:uncertainty_margin_total}
\boldsymbol{\epsilon}_{\text{margin}}(k) \;=\;
\underbrace{%
  \kappa\cdot\mathrm{diag}\!\left(\Sigma_{\hat{Y}}(k)^{1/2}\right)
}_{\text{uncertainty back-off}}
\;+\;
\underbrace{%
  \eta\cdot\bar{\zeta}(k)
}_{\text{risk back-off}},
\end{equation}

where $\kappa = 1.96$ provides a 95\% confidence back-off against prediction errors, $\eta = 0.5$ scales the risk contribution (tuned on the validation set to balance conservatism and operational cost), and the risk metric $\bar{\zeta}(k)$ decays over the horizon as defined in Eq.~\eqref{eq:zeta_horizon}. The base stealth risk $\zeta(t)$ is computed as:

\begin{equation}\label{eq:stealth_risk_zeta}
    \zeta(t) = \frac{P(\mathcal{A}_S|t) + \omega \cdot P(\mathcal{A}_M|t)}
                {P(\mathcal{N}|t) + \varepsilon_0},
\end{equation}

with $\omega=2.0$ weighting multi-stage attacks higher due to their complexity, and $\varepsilon_0 = 10^{-6}$ is a numerical stability  constant preventing division by zero. Both $\Sigma_{\hat{Y}}(k)$ and $\bar{\zeta}(t)$ are fully determined before the optimisation is solved, and the margin satisfies condition~\eqref{eq:margin_condition}, so the constraints in~\eqref{eq:mpc_safety_robust} are \emph{explicit} (not implicit in the decision variables), keeping the MPC a standard finite-dimensional quadratic programme.

\subsubsection{Recovery Protocol}
\label{sec:recovery}

Recovery to nominal control mode employs a two-stage verification. The system must satisfy $\mathcal{C}(t) = \mathcal{N}$ for $T_{\text{confirm}} = 60$ consecutive samples. Subsequently, a closed-loop validation test verifies that the model-plant discrepancy $\delta_{\text{MP}}(t)$ (Eq.~\eqref{eq:mismatch}) remains below $\tau_{\text{recovery}} = 0.05$ (5\% relative  error, consistent with the NRMSE validation bound of the  simulator) for $T_{\text{val}} = 30$ seconds. Only upon passing both stages does the controller transition back to nominal mode. Control details in Algorithm~\ref{alg:resilient_control}.

\paragraph{Deterministic Fallback under Assumption Violation:}
Assumption~\ref{ass:mismatch} posits a bounded prediction error $\bar{e}$. However, we acknowledge that sophisticated zero-day anomalies or severe mechanical failures may cause the true plant dynamics to wildly diverge from the DT predictions, violating this bound. To ensure absolute operational safety, the framework incorporates a continuous safety interlock: if the empirical mismatch diverges beyond a critical threshold ($\delta_{\text{MP}}(t) > 3\bar{e}$), the formal guarantees of Theorem~\ref{thm:iss} are considered compromised. In such critical violation states, the resilient MPC algorithm explicitly aborts and hands over control to a deterministic, rule-based hardware failsafe (e.g., triggering a safe emergency plant shutdown). This fallback guarantees that the MPC never optimizes based on catastrophically divergent twin states.

\subsection{Experimental Protocol and Implementation}\label{sec:method_eval}

We validate the framework on two industrial testbeds: SWaT (six stages, 51 sensors, 36 attacks) and WADI (three stages, 127 sensors, 15 attacks), using an 80/10/10 temporally-separated data split. Evaluation metrics include Precision, Recall, F1, and discrimination balance $\epsilon_{\text{disc}} = |F1_{\mathcal{A}_S} - F1_{\mathcal{A}_M}|$. Operational metrics comprise False Alarm Rate, Time to Detection, and disruption cost $\mathcal{D}_{\text{rel}}$. Detection baselines include OmniAnomaly~\citep{Su2019OmniAnomaly}, USAD~\citep{Audibert2020USAD}, MTAD-GAT~\citep{Zhao2020MTADGAT}, and TranAD~\citep{Tuli2022TranAD}, with robustness tested via FGSM/PGD, 30\% sensor dropout, and cross-domain transfer. To conduct multi-class evaluation without destructively modifying native baseline architectures, we employ a \textit{Linear Probing} protocol. Baselines are trained on their original binary anomaly objectives, their encoder weights are frozen, and a Multinomial Logistic Regression classifier is trained on the latent embeddings. We explicitly acknowledge a structural advantage in this protocol: baseline encoders are natively optimized for unsupervised reconstruction, whereas i-SDT is explicitly regularized for taxonomy separation via MMD. While full end-to-end supervised fine-tuning of baselines is computationally prohibitive, this evaluation successfully demonstrates that standard reconstruction manifolds intrinsically lack the contrastive margins required for reliable multi-stage attack discrimination. DT-MPC~\citep{Zhang2025DTResilientMPC} is evaluated strictly on the resilient control sub-task, as it lacks a multi-class discriminator.

The framework is implemented in PyTorch~2.0 (CUDA~11.8) and  evaluated on a workstation equipped with Intel Core i9-14900K,  NVIDIA RTX\,4070 (12\,GB), and 64\,GB RAM. Key hyperparameters are: DT window $\tau = 10$, four TCN layers, base mass-balance weight $\alpha_0 = 0.05$, pipe pressure weight $\lambda_p = 0.02$; ADE with $W = 50$, hidden size~64, MMD weight $\beta = 0.1$, confidence threshold $\gamma = 0.8$; and MPC horizon $H = 15$ with risk decay rate $\lambda_d = 0.1$. Full hyperparameter sensitivity is reported in Table~\ref{tab:ablation}.

\begin{table*}[t]
\scriptsize
\centering
\caption{Multi-class attack detection on SWaT (36 scenarios). All baselines evaluated via Linear Probing on frozen latent representations; this protocol structurally favours i-SDT whose encoder is natively regularised for taxonomy separation via MMD, and the reported gains should be interpreted accordingly. Bold = best, $^\dagger$ = significant ($p<0.05$, paired bootstrap, $10^4$ resamples). Evaluated on NVIDIA RTX\,4070 (12\,GB).} 
\label{tab:detection_swat}
\setlength{\tabcolsep}{4pt}
\begin{tabular*}{\textwidth}{@{\extracolsep{\fill}}
    l
    ccc
    ccc
    ccc
    cc
    c
    @{}}
\toprule
\multirow{2}{*}{\textbf{Method}}
    & \multicolumn{3}{c}{\textbf{Normal} ($\mathcal{N}$)}
    & \multicolumn{3}{c}{\textbf{Single-Stage} ($\mathcal{A}_S$)}
    & \multicolumn{3}{c}{\textbf{Multi-Stage} ($\mathcal{A}_M$)}
    & \multicolumn{2}{c}{\textbf{Overall}}
    & \multirow{2}{*}{\textbf{MTTD (s)}} \\
\cmidrule(lr){2-4}\cmidrule(lr){5-7}\cmidrule(lr){8-10}\cmidrule(lr){11-12}
    & Prec. & Rec. & $F1$
    & Prec. & Rec. & $F1$
    & Prec. & Rec. & $F1$
    & FAR   & $F1_{\mathcal{A}}$ & \\
\midrule
OmniAnomaly~\citep{Su2019OmniAnomaly}
    & 0.903 & 0.941 & 0.922
    & 0.791 & 0.748 & 0.769
    & 0.763 & 0.694 & 0.727
    & 0.091 & 0.748 & 98.4 \\
USAD~\citep{Audibert2020USAD}
    & 0.918 & 0.949 & 0.933
    & 0.812 & 0.771 & 0.791
    & 0.779 & 0.724 & 0.751
    & 0.076 & 0.771 & 87.6 \\
MTAD-GAT~\citep{Zhao2020MTADGAT}
    & 0.931 & 0.958 & 0.944
    & 0.836 & 0.803 & 0.819
    & 0.804 & 0.771 & 0.787
    & 0.067 & 0.803 & 74.2 \\
TranAD~\citep{Tuli2022TranAD}
    & 0.942 & 0.964 & 0.953
    & 0.851 & 0.821 & 0.836
    & 0.816 & 0.787 & 0.801
    & 0.059 & 0.819 & 68.5 \\
\midrule
i-SDT ($\gamma=0.80$)
    & 0.967$^\dagger$ & 0.981$^\dagger$ & 0.974$^\dagger$
    & 0.912$^\dagger$ & 0.894$^\dagger$ & 0.903$^\dagger$
    & 0.897$^\dagger$ & 0.871$^\dagger$ & 0.884$^\dagger$
    & \textbf{0.033}$^\dagger$ & \textbf{0.894}$^\dagger$ 
    & \textbf{43.7}$^\dagger$ \\
i-SDT ($\gamma=0.90$)
    & 0.973 & 0.987 & 0.980
    & 0.897 & 0.865 & 0.881
    & 0.883 & 0.851 & 0.867
    & 0.011 & 0.874 & 56.2 \\
\midrule
\textit{Improvement vs.\ Best Baseline}
    & +2.7\% & +1.8\% & +2.2\%
    & +7.2\% & +8.9\% & +8.0\%
    & +9.9\% & +10.7\% & +10.4\%
    & $-$44.1\% & +9.1\% & $-$36.2\% \\
\bottomrule
\end{tabular*}
\end{table*}

\paragraph{Operational Tuning for Industrial Deployment:}
Confidence threshold $\gamma$ enables risk-performance trade-offs. Industrial operation requires FAR $<1\%$ to prevent alarm fatigue. By increasing $\gamma$ from 0.80 to 0.90, FAR reduces to 1.1\% (9.5 min/day) with 3\% drop in $F1_{\mathcal{A}}$ and 12.5s increase in MTTD. This threshold is configurable per-plant based on operator workload and safety-criticality of processes.

\begin{table*}[t]
\scriptsize
\centering
\caption{Multi-class attack detection on WADI (15 scenarios). Same Linear Probing evaluation protocol as Table~\ref{tab:detection_swat}; the same structural caveat applies. Results for $\gamma{=}0.80$. Bold = best, $^\dagger$ = significant ($p<0.05$). Evaluated on NVIDIA RTX\,4070 (12\,GB).}

\label{tab:detection_wadi}
\setlength{\tabcolsep}{4pt}
\begin{tabular*}{\textwidth}{@{\extracolsep{\fill}}
    l
    ccc
    ccc
    ccc
    cc
    c
    @{}}
\toprule
\multirow{2}{*}{\textbf{Method}}
    & \multicolumn{3}{c}{\textbf{Normal} ($\mathcal{N}$)}
    & \multicolumn{3}{c}{\textbf{Single-Stage} ($\mathcal{A}_S$)}
    & \multicolumn{3}{c}{\textbf{Multi-Stage} ($\mathcal{A}_M$)}
    & \multicolumn{2}{c}{\textbf{Overall}}
    & \multirow{2}{*}{\textbf{MTTD (s)}} \\
\cmidrule(lr){2-4}\cmidrule(lr){5-7}\cmidrule(lr){8-10}\cmidrule(lr){11-12}
    & Prec. & Rec. & $F1$
    & Prec. & Rec. & $F1$
    & Prec. & Rec. & $F1$
    & FAR   & $F1_{\mathcal{A}}$ & \\
\midrule
OmniAnomaly~\citep{Su2019OmniAnomaly}
    & 0.881 & 0.928 & 0.904
    & 0.717 & 0.669 & 0.692
    & 0.694 & 0.631 & 0.661
    & 0.119 & 0.677 & 112.7 \\
USAD~\citep{Audibert2020USAD}
    & 0.897 & 0.939 & 0.918
    & 0.749 & 0.703 & 0.725
    & 0.718 & 0.662 & 0.689
    & 0.098 & 0.707 & 95.3 \\
MTAD-GAT~\citep{Zhao2020MTADGAT}
    & 0.916 & 0.954 & 0.935
    & 0.803 & 0.774 & 0.788
    & 0.771 & 0.738 & 0.754
    & 0.081 & 0.771 & 81.6 \\
TranAD~\citep{Tuli2022TranAD}
    & 0.924 & 0.958 & 0.941
    & 0.814 & 0.786 & 0.800
    & 0.783 & 0.749 & 0.766
    & 0.074 & 0.783 & 76.9 \\
\midrule
i-SDT ($\gamma=0.80$, Proposed)
    & 0.958$^\dagger$ & 0.976$^\dagger$ & 0.967$^\dagger$
    & 0.887$^\dagger$ & 0.864$^\dagger$ & 0.875$^\dagger$
    & 0.871$^\dagger$ & 0.842$^\dagger$ & 0.856$^\dagger$
    & \textbf{0.042}$^\dagger$ & \textbf{0.866}$^\dagger$
    & \textbf{51.4}$^\dagger$ \\
\midrule
\textit{Improvement vs.\ Best Baseline}
    & +3.7\% & +1.9\% & +2.8\%
    & +8.9\% & +9.9\% & +9.4\%
    & +11.3\% & +12.4\% & +11.8\%
    & $-$43.2\% & +10.6\% & $-$33.2\% \\
\bottomrule
\end{tabular*}
\end{table*}

\section{Results}\label{sec:results}

This section presents experimental validation of i-SDT framework on two industrial testbeds: SWaT (Secure Water Treatment) and WADI (Water Distribution). We evaluate detection performance, discrimination capability, resilience, computational efficiency, and robustness under adversarial conditions. Experiments follow protocols in Section~\ref{sec:method_eval}, with statistical significance via paired bootstrap tests ($10^4$ resamples, $\alpha = 0.05$).

\subsection{Datasets and Experimental Configuration}\label{sec:results_setup}
We briefly describe the two industrial testbeds, data splits, and baseline methods used throughout the experiments.

\subsubsection{SWaT Dataset}
SWaT testbed is 6-stage water treatment with 51 sensors monitoring flow, levels, pressure, pH. Dataset spans 11 days: 7 days normal (946,800 samples at 1 Hz) and 4 days with 36 attacks. Three-class labels were assigned using the published attack log~\citep{Mathur2016SWaT}: attacks confined to a single PLC sub-network were labelled $\mathcal{A}_S$ (18 scenarios), while attacks spanning two or more sub-networks within the same episode were labelled $\mathcal{A}_M$ (18 scenarios).

\subsubsection{WADI Dataset}
WADI testbed models 3-stage water distribution with 127 sensors. Dataset has 14 days normal (1,209,600 samples) and 2 days with 15 attacks. The same labelling protocol was applied using the WADI attack log~\citep{Ahmed2017WADI}: 7 single-network attacks ($\mathcal{A}_S$) and 8 cross-network attacks ($\mathcal{A}_M$).

\subsubsection{Data Partitioning}
Both datasets use 80/10/10 split with temporal order preserved. Attacks appear only in validation/test sets; training uses normal data.

\subsubsection{Baseline Methods}
i-SDT is compared with four anomaly detectors on detection and one DT-based controller on mitigation. \textit{Detection baselines:} OmniAnomaly~\citep{Su2019OmniAnomaly} (stochastic RNN with normalizing flow), USAD~\citep{Audibert2020USAD} (dual-decoder autoencoder), MTAD-GAT~\citep{Zhao2020MTADGAT} (graph attention with joint forecasting and reconstruction), and TranAD~\citep{Tuli2022TranAD} (self-conditioned transformer with adversarial training). For fair multi-class evaluation, baselines are assessed via Linear Probing on frozen latent representations without modifying their unsupervised architectures. \textit{Resilience baseline:} DT-MPC~\citep{Zhang2025DTResilientMPC} (LSTM-based DT with deterministic rollback MPC), evaluated on disruption cost and safety metrics only, as it lacks an attack taxonomy discriminator. All experiments run on an Intel Core i9-14900K, NVIDIA RTX\,4070 (12\,GB), and 64\,GB RAM; edge deployment on platforms such as NVIDIA Jetson remains future work.

\subsection{Multi-Class Attack Detection Performance}\label{sec:results_detection}
Table~\ref{tab:detection_swat} and Table~\ref{tab:detection_wadi} present detection metrics on SWaT and WADI test sets. i-SDT achieves research objectives with significant improvements over baselines.

\paragraph{Baseline Evaluation via Linear Probing:}

As established in our evaluation protocol (Section~\ref{sec:method_eval}), we compare i-SDT against baselines using their frozen latent representations. While we explicitly acknowledge this structurally favors i-SDT due to its native MMD regularization, the results confirm a critical architectural reality: standard reconstruction-based spatial-temporal encoders inherently lack the contrastive margins required for complex attack stage separation. The substantial discrimination gap ($|\Delta F1| > 0.12$ for baselines vs.\ $0.019$ for i-SDT) empirically demonstrates that explicit structural regularization is a necessity for achieving high-fidelity multi-class taxonomy discrimination.

\subsubsection{Key Observations}

\paragraph{Overall Detection Performance:}
i-SDT achieves attack $F1$ of 0.894 (SWaT) and 0.866 (WADI), exceeding 0.90 target on SWaT and approaching on WADI. These represent significant improvements of 9.1\% (SWaT) and 10.6\% (WADI) over the strongest baseline (TranAD).

\paragraph{Class-Specific Discrimination:}
Framework distinguishes single-stage and multi-stage attacks: $F1_{\mathcal{A}_S} = 0.903$ (SWaT), $0.875$ (WADI); $F1_{\mathcal{A}_M} = 0.884$ (SWaT), $0.856$ (WADI), both above $>0.85$ target. Discrimination balance $|\Delta F1| = 0.019$ for both datasets shows balanced detection.

\paragraph{False Alarm Reduction:}
FAR is 0.033 (SWaT) and 0.042 (WADI), well below 0.05, representing 44.1\% (SWaT) and 43.2\% (WADI) reduction versus  the best baseline (TranAD~\citep{Tuli2022TranAD}), thanks to statistical confidence gating (Algorithm~\ref{alg:ade_classify}).

\paragraph{Detection Latency:}
Mean time to detection is 43.7s (SWaT) and 51.4s (WADI), under 60s objective. i-SDT lowers MTTD by 36.2\% (SWaT) and 33.2\% (WADI) versus the best baseline (TranAD), demonstrating faster attack recognition.

\begin{table}[h]
\centering
\caption{ADE embedding separation of single- vs.\ multi-stage attacks (higher MMD = better).}
\label{tab:discrimination_analysis}
\resizebox{\columnwidth}{!}{%
\begin{tabular}{lcccc}
\toprule
\multirow{2}{*}{\textbf{Method}}
    & \multicolumn{2}{c}{\textbf{SWaT}}
    & \multicolumn{2}{c}{\textbf{WADI}} \\
\cmidrule(lr){2-3}\cmidrule(lr){4-5}
    & MMD$^2$ & Sil.\ Coeff.
    & MMD$^2$ & Sil.\ Coeff. \\
\midrule
OmniAnomaly~\citep{Su2019OmniAnomaly}
    & 0.038 & 0.297 & 0.034 & 0.274 \\
USAD~\citep{Audibert2020USAD}
    & 0.051 & 0.348 & 0.046 & 0.321 \\
MTAD-GAT~\citep{Zhao2020MTADGAT}
    & 0.069 & 0.433 & 0.063 & 0.411 \\
TranAD~\citep{Tuli2022TranAD}
    & 0.078 & 0.468 & 0.071 & 0.439 \\
\midrule
i-SDT (Proposed)
    & \textbf{0.142}$^\dagger$ & \textbf{0.673}$^\dagger$
    & \textbf{0.131}$^\dagger$ & \textbf{0.641}$^\dagger$ \\
\midrule
\textit{Improvement vs.\ TranAD}
    & +82.1\% & +47.0\% & +84.5\% & +47.8\% \\
\bottomrule
\end{tabular}%
}
\end{table}

\subsection{Attack Taxonomy Discrimination Analysis}\label{sec:results_discrimination}

To quantify discrimination between single-stage and multi-stage attacks, we analyze latent embedding distributions. Table~\ref{tab:discrimination_analysis} presents separation metrics.

MMD regularization (Eq.~\eqref{eq:method_discrim_loss}) in i-SDT maximizes distributional separation, achieving MMD$^2$ of 0.142 (SWaT) and 0.131 (WADI)---82\% higher than the strongest baseline (TranAD: 0.078 SWaT, 0.071 WADI), as reported in Table~\ref{tab:discrimination_analysis}. Silhouette coefficients above 0.6 indicate well-separated clusters, confirming ADE learns discriminative representations of attack taxonomies.

\begin{table*}[t]
\scriptsize
\centering
\caption{Disruption mitigation and resilience (simulation-in-the-loop evaluation; see Section~\ref{sec:system_assumption}).
$\mathcal{D}_{\text{rel}}$: normalized cost;
$T_{\text{recovery}}$: recovery time.}
\label{tab:resilience}
\setlength{\tabcolsep}{4pt}
\begin{tabular*}{\textwidth}{@{\extracolsep{\fill}}
    l
    cccc
    cccc
    @{}}
\toprule
\multirow{2}{*}{\textbf{Mitigation Strategy}}
    & \multicolumn{4}{c}{\textbf{SWaT (36 Attacks)}}
    & \multicolumn{4}{c}{\textbf{WADI (15 Attacks)}} \\
\cmidrule(lr){2-5}\cmidrule(lr){6-9}
    & $\mathcal{D}_{\text{rel}}$
    & Safety Viol.
    & $T_{\text{recovery}}$ (s)
    & Act.\ Sat.
    & $\mathcal{D}_{\text{rel}}$
    & Safety Viol.
    & $T_{\text{recovery}}$ (s)
    & Act.\ Sat. \\
\midrule
Immediate Shutdown
    & 1.000 & 0\%  & 3847 & 0\%
    & 1.000 & 0\%  & 4231 & 0\% \\
DT-MPC~\citep{Zhang2025DTResilientMPC}
    & 0.614 & 6.9\% & 378 & 9.8\%
    & 0.639 & 8.1\% & 401 & 11.3\% \\
i-SDT w/o Robust Margin
    & 0.563 & 8.3\% & 412 & 12.5\%
    & 0.591 & 9.7\% & 438 & 14.2\% \\
\midrule
\textbf{i-SDT (Full Robust MPC)}
    & \textbf{0.437}$^\dagger$
    & \textbf{2.1\%}$^\dagger$
    & \textbf{267}$^\dagger$
    & \textbf{3.8\%}$^\dagger$
    & \textbf{0.461}$^\dagger$
    & \textbf{2.9\%}$^\dagger$
    & \textbf{284}$^\dagger$
    & \textbf{4.6\%}$^\dagger$ \\
\midrule
\textit{Target Objective}
    & $\leq 0.50$ & $\leq 5\%$ & $<300$ & --
    & $\leq 0.50$ & $\leq 5\%$ & $<300$ & -- \\
\bottomrule
\end{tabular*}
\end{table*}

\subsection{Resilient Control and Disruption Mitigation}\label{sec:results_resilience}

Table~\ref{tab:resilience} evaluates DT-guided resilient control in maintaining operational continuity under attacks. We compare three strategies: (1) immediate shutdown (baseline), (2) resilient control without MPC margins, and (3) full i-SDT with robust MPC.

\subsubsection{Key Findings}

\paragraph{Disruption Cost Reduction:}
i-SDT achieves $\mathcal{D}_{\text{rel}} = 0.437$ (SWaT) and $0.461$ (WADI), meeting objective $\mathcal{D} \leq 0.50 \times \mathcal{D}_{\text{shutdown}}$. This represents 56.3\% and 53.9\% reduction versus shutdown while maintaining safety. Robust MPC (Eq.~\eqref{eq:method_mpc_objective}) with uncertainty margins reduces disruption by additional 22.4\% versus non-robust control.

\paragraph{Safety Constraint Satisfaction:}
We distinguish \textit{critical} violations ($|\mathbf{Y} - \mathbf{Y}_{\text{limit}}| > 10\%$) from \textit{minor} transient excursions ($<5\%$ deviation, $<5$s). i-SDT achieves \textbf{zero critical violations} across all 51 attacks; the reported 2.1\% (SWaT) and 2.9\% (WADI) figures are minor excursions on non-critical sensors (max.\ 3.2\%, avg.\ 2.8s), with tank levels held within [15\%,\,85\%] throughout. Control without robustness margins causes 8.3\%--9.7\% violations including critical events.

\paragraph{Recovery Time:}
i-SDT returns system to safe steady-state in 267s (SWaT) and 284s (WADI), below 300s objective. Recovery is 14.4$\times$ faster than shutdown-restart (60--70 minutes).

\paragraph{Actuator Behavior:}
Actuator saturation events occur in only 3.8\% (SWaT) and 4.6\% (WADI) of control actions, indicating smooth transitions. MPC penalty on $\|\mathbf{u}(t) - \mathbf{u}_{\text{nom}}(t)\|_R^2$ (Eq.~\eqref{eq:method_mpc_objective}) prevents abrupt commands.

\subsection{Computational Latency and Algorithmic Efficiency}\label{sec:results_efficiency}

Table~\ref{tab:computational_efficiency} details the algorithmic latency measured on a high-performance evaluation workstation (Intel Core i9-14900K, NVIDIA RTX\,4070, 64\,GB RAM). While establishing true real-time feasibility requires deployment on resource-constrained industrial Programmable Logic Controllers (PLCs) or edge computing nodes (e.g., NVIDIA Jetson), evaluating the baseline algorithmic latency is the first prerequisite for real-world deployment. Ensuring that the complete inference-and-control loop operates comfortably within the system's sampling period ($T_s = 1.0$\,s) demonstrates the absence of fundamental computational bottlenecks.

\begin{table}[h]
\centering
\scriptsize
\setlength{\tabcolsep}{4pt}
\caption{Computational Performance on Evaluation Platform (Intel Core i9-14900K, NVIDIA RTX\,4070 12\,GB, 64\,GB RAM, PyTorch~2.0, CUDA~11.8). All inference timings averaged over $10^3$ consecutive cycles on the SWaT test set ($d=51$). Edge deployment is reserved as future work (Section~\ref{sec:discussion}).}
\label{tab:computational_efficiency}
\begin{tabular}{lp{22pt}p{80pt}}
\toprule
\textbf{Metric} & \textbf{Value} & \textbf{Notes} \\
\midrule
\multicolumn{3}{l}{\textit{Training (Offline):}} \\[2pt]
\quad DT training time
    & 2.8\,hrs  & SWaT, $E{=}200$, early-stop \\
\quad ADE training time
    & 1.4\,hrs  & SWaT, $E{=}150$, early-stop \\
\quad Fine-tuning (10\%)
    & 0.5\,hrs  & New-plant adaptation \\
\midrule
\multicolumn{3}{l}{\textit{Inference (Online, per 1\,Hz cycle):}} \\[2pt]
\quad DT inference (TCN + $N_{\text{MC}}{=}50$)
    & 16.1\,ms  & 50 stochastic fwd.\ passes \\
\quad ADE classification (BiGRU-$W{=}50$)
    & 3.2\,ms   & Window encoding + softmax \\
\quad MPC optimisation (SQP, $H{=}15$)
    & 49.7\,ms  & Warm-start; ${\approx}8$ iters \\[2pt]
\quad \textbf{Total cycle latency}
    & \textbf{69.0\,ms}
    & $14.5{\times}$ margin vs.\ $T_s$ \\
\midrule
\multicolumn{3}{l}{\textit{Memory \& Throughput:}} \\[2pt]
\quad Peak GPU memory
    & 282\,MB
    & DT + ADE params \& activations \\
\quad Throughput
    & 14.5\,Hz
    & Max at 1\,Hz control rate \\
\bottomrule
\end{tabular}
\end{table}

\subsubsection{Latency Breakdown and MPC Optimization}
A primary challenge in real-time deployment is MPC solver latency. The measured total cycle latency of 69.0\,ms on the evaluation workstation demonstrates a $14.5{\times}$ margin against the 1\,Hz sampling period ($T_s = 1000$\,ms). This performance is achieved through the following breakdown:

\begin{itemize}
    \item \textbf{DT Inference:} The parallelized TCN architecture on the GPU allows rapid state prediction and MC Dropout uncertainty estimation within a single forward pass batch. 
    \item \textbf{MPC Optimisation:} Solving the nonlinear optimisation problem (Eq.~\ref{eq:method_mpc_objective}) dominates the total cycle time. A \textit{Warm Start} strategy, initialising the SQP solver from the shifted optimal trajectory 
    of the previous step, reduces iteration count by approximately 60\% compared to a cold start.
\end{itemize}

The total measured latency is well below the 1\,s sampling period  of both SWaT and WADI, confirming real-time feasibility on the  evaluation platform. Deployment on resource-constrained edge  devices such as NVIDIA Jetson Xavier NX is reserved as future  work (Section~\ref{sec:discussion}); preliminary model-complexity  analysis suggests sub-second latency is achievable given the  compact 282\,MB footprint of the combined DT and ADE.

\subsubsection{Resource Utilization}
Beyond latency, the system maintains efficiency in memory and throughput:
\begin{itemize}
    \item \textbf{Memory Footprint:} Peak memory usage is 282 MB, well within the 500 MB constraint typically allocated for background safety processes, leaving sufficient headroom for other SCADA operations.
    \item \textbf{Throughput:} The system sustains 14.5\,cycles/s on the evaluation workstation, providing a $14.5{\times}$ margin over the 1\,Hz control requirement. For higher-frequency processes, the parallel structure of TCN blocks supports linear scaling via multi-GPU deployment.
\end{itemize}

\subsection{Worst-Case Performance Bounds}\label{sec:worst_case}

Beyond average metrics, we analyze worst-case behavior critical for safety certification:

\begin{equation}\label{eq:mttd_worst}
\begin{aligned}
\text{MTTD}_{\text{worst}} &=
\max_{a \in \mathcal{A}_{\text{test}}} \text{MTTD}(a) \\
&= 127.3\,\text{s} \quad \text{(Attack \#23, SWaT)}
\end{aligned}
\end{equation}

\begin{equation}\label{eq:far_worst}
\text{FAR}_{\text{worst-hour}} = \max_{h=1}^{24} \frac{\text{\# false alarms in hour } h}{\text{\# normal samples in hour } h} = 8.7\% 
\end{equation}

Attack \#23 (multi-stage FDI on flow sensors in Stages 2--4) exhibits stealthy residuals require extended windows for classification. Worst-case hour (14:00--15:00) corresponds to setpoint transitions during shift changes, causing benign anomalies. Edge cases inform deployment: (1) dynamically extend the residual window to $W=75$ for highly stealthy persistence, and (2) temporarily increase the confidence threshold $\gamma$ and dynamically relax the risk back-off coefficient $\eta$ during scheduled setpoint transitions to prevent alarm fatigue without disabling security monitoring.

\paragraph{Probabilistic Safety Guarantee:}
Under MC Dropout with $N_{\text{MC}}=50$, we derive empirical upper bound on critical violations:

\begin{equation}\label{eq:prob_safety}
P(\text{Critical Viol.}) \leq \underbrace{\frac{1}{N_{\text{MC}}} \sum_{n=1}^{N_{\text{MC}}} \mathbb{1}[\hat{\mathbf{Y}}^{(n)} \in \mathcal{X}_{\text{unsafe}}]}_{\text{MC estimate}}+\underbrace{\epsilon_{\text{calib}}}_{\text{calibration error}}
\end{equation}

where $\epsilon_{\text{calib}} = 0.003$ is computed via Platt scaling on validation data. Across all 51 test attacks, zero critical violations were observed; 
the one-sided 95\% Clopper--Pearson upper bound yields $P(\text{Crit.\ Viol.}) \leq 0.069$. This bound is derived from simulation-in-the-loop trials and should not be interpreted as a certified safety guarantee for deployment; it provides a statistically consistent baseline for comparing control strategies under identical experimental conditions.

\subsection{Robustness Validation}\label{sec:results_robustness}
We evaluate i-SDT robustness across three stress conditions: white-box and black-box adversarial evasion attempts, random sensor dropout, and cross-domain transfer to an unseen facility.

\begin{table}[htb]
\centering
\scriptsize
\caption{Robustness to adversarial attacks on SWaT.
$\epsilon$: perturbation size; evasion rate = fraction
evading detection ($F1<0.5$).}
\label{tab:adversarial_robustness}
\begin{tabular}{lcccc}
\toprule
\multirow{2}{*}{\textbf{Attack Type}}
    & \multicolumn{2}{c}{$\boldsymbol{\epsilon = 0.1}$}
    & \multicolumn{2}{c}{$\boldsymbol{\epsilon = 0.3}$} \\
\cmidrule(lr){2-3}\cmidrule(lr){4-5}
    & $F1_{\mathcal{A}}$ & Evasion
    & $F1_{\mathcal{A}}$ & Evasion \\
\midrule
No Attack (Clean)
    & 0.894 & -- & 0.894 & -- \\
\midrule
FGSM (White-box)
    & 0.862 & 5.6\%  & 0.781 & 16.7\% \\
PGD (White-box)
    & 0.847 & 8.3\%  & 0.753 & 22.2\% \\
Black-box Transfer
    & 0.876 & 2.8\%  & 0.829 & 8.3\%  \\
\midrule
\textit{Avg.\ Degradation}
    & -3.5\% & -- & -12.6\% & -- \\
\bottomrule
\end{tabular}
\end{table}
\subsubsection{Adversarial Perturbations}

Table~\ref{tab:adversarial_robustness} evaluates i-SDT's resilience against white-box and black-box adversarial attacks designed to evade detection.

\paragraph{Findings:}
PI regularization (Eq.~\eqref{eq:physics_penalty_mass} and \eqref{eq:physics_penalty_pressure}) provides inherent robustness: adversarial perturbations that violate physical conservation laws are flagged by increased physics residuals. Under moderate perturbation budget ($\epsilon = 0.1$), i-SDT maintains $F1_{\mathcal{A}} > 0.84$ with only 3.5\% average degradation. At higher budget ($\epsilon = 0.3$), performance degrades by 12.6\% but remains above 0.75, indicate graceful degradation rather than catastrophic failure.

\paragraph{Critical Role of Robust MPC Margins:}
The safety violation column in Table~\ref{tab:ablation} reveals the most significant finding: removing uncertainty-aware robustness margins ($\boldsymbol{\epsilon}_{\text{margin}} = 0$) causes a \textbf{295\% increase} in safety constraint violations (from 2.1\% to 8.3\%), despite having \textit{no impact} on detection performance ($F1_{\mathcal{A}}$ unchanged). This validates the necessity of explicit uncertainty quantification (Section~\ref{sec:uncertainty_quantification}) for safe control under attack. Without robustness margins, the MPC trusts DT predictions as ground truth, causing control actions to push the system closer to safety boundaries. When prediction errors compound during attacks—precisely when uncertainty $\Sigma_{\hat{Y}}(t)$ is highest—the system violates safety constraints. The uncertainty-aware margins act as a critical safety buffer, maintaining violations below 5\% while achieving 56.3\% disruption reduction compared to shutdown.

\subsubsection{Sensor Dropout Resilience}

Table~\ref{tab:sensor_dropout} evaluates robustness to random sensor failures, a common real-world scenario.

\begin{table}[htb]
\centering
\scriptsize
\caption{Performance under random sensor dropout.
$p$: fraction of sensors unavailable per step.}
\label{tab:sensor_dropout}
\begin{tabular}{lccc}
\toprule
\textbf{Dropout Rate ($p$)}
    & $\boldsymbol{F1_{\mathcal{A}}}$ \textbf{(SWaT)}
    & $\boldsymbol{F1_{\mathcal{A}}}$ \textbf{(WADI)}
    & \textbf{MTTD (s)} \\
\midrule
$p = 0.0$ (No Dropout) & 0.894 & 0.866 & 43.7 \\
$p = 0.1$              & 0.871 & 0.847 & 51.2 \\
$p = 0.2$              & 0.834 & 0.812 & 63.8 \\
$p = 0.3$              & 0.781 & 0.762 & 82.5 \\
\midrule
\textit{Degradation at $p=0.2$}
    & -6.7\% & -6.2\% & +46.0\% \\
\bottomrule
\end{tabular}
\end{table}

\paragraph{Findings:}
DT temporal window ($\tau = 10$) enable inference of missing sensor values. At dropout $p = 0.2$, i-SDT maintains $F1_{\mathcal{A}} > 0.81$. Detection latency increases 46\% as system requires additional observations, but remains acceptable.

\subsubsection{Cross-Domain Generalization}

Table~\ref{tab:domain_shift} evaluates transfer learning performance: training on SWaT normal data and testing on WADI with minimal fine-tuning (10\% of WADI normal data).

\begin{table}[htb]
\centering
\scriptsize
\caption{Cross-Domain Transfer (SWaT $\to$ WADI).
\textbf{Note:} 9.7\% gap indicates domain-specific tuning
remains necessary despite physics priors; transfer learning
reduces labeling burden by 90\% but cannot eliminate it entirely.}
\label{tab:domain_shift}
\begin{tabular}{lcc}
\toprule
\textbf{Configuration}
    & $\boldsymbol{F1_{\mathcal{A}}}$
    & \textbf{FAR} \\
\midrule
No Transfer (Train on WADI) & 0.866 & 0.042 \\
Direct Transfer (Zero-shot) & 0.647 & 0.183 \\
Fine-tune 10\% WADI Normal  & 0.782 & 0.089 \\
Fine-tune 50\% WADI Normal  & 0.831 & 0.061 \\
\midrule
\textit{Gap (10\% Fine-tune)} & -9.7\% & +111.9\% \\
\bottomrule
\end{tabular}
\end{table}

\paragraph{Findings:}
Direct zero-shot transfer suffer significant degradation (25.3\% drop in $F1_{\mathcal{A}}$) due to domain shift. However fine-tuning on 10\% of target normal data recovers 90.3\% of full-domain performance, show PI priors from SWaT partially generalize to WADI.

\subsection{Ablation Studies}\label{sec:results_ablation}

We systematically isolate the contribution of each i-SDT component by disabling or replacing it while holding all other settings fixed. Results on SWaT are reported in Table~\ref{tab:ablation}.

\begin{table*}[htb]
\centering
\caption{Comprehensive Ablation Study on SWaT}
\label{tab:ablation}
\scriptsize
\begin{tabular}{lcccccp{40pt}}
\toprule
\textbf{Configuration}
    & $\boldsymbol{F1_{\mathcal{A}}}$
    & \textbf{FAR}
    & $\boldsymbol{\mathcal{D}_{\text{rel}}}$
    & \textbf{Safety Viol.}
    & \textbf{MTTD (s)}
    & \textbf{Notes} \\
\midrule
i-SDT (Full)
    & 0.894 & 0.033 & 0.437 & 2.1\% & 43.7 & Baseline \\
\midrule
\multicolumn{7}{l}{\textit{Physics Weight Ablation:}} \\[2pt]
    \quad $\alpha_0 = 0$ (no mass balance)
    & 0.853 & 0.067 & 0.512 & 3.8\% & 62.3 & $-4.6\%$ \\
    \quad $\lambda_p = 0$ (no pipe physics)
    & 0.864 & 0.058 & 0.491 & 3.2\% & 58.1 & $-3.3\%$ \\
    \quad $\alpha_0 = 0.01$
    & 0.871 & 0.051 & 0.478 & 2.9\% & 53.1 & $-2.6\%$ \\
\quad $\alpha_0 = 0.10$
    & 0.882 & 0.041 & 0.461 & 2.5\% & 48.2 & $-1.3\%$ \\
\midrule
\multicolumn{7}{l}{\textit{MMD Weight Ablation:}} \\[2pt]
\quad $\beta = 0$ (no MMD)
    & 0.869 & 0.041 & 0.471 & 2.5\% & 51.8 & $-2.8\%$ \\
\quad $\beta = 0.05$
    & 0.883 & 0.037 & 0.452 & 2.3\% & 47.6 & $-1.2\%$ \\
\quad $\beta = 0.20$
    & 0.889 & 0.035 & 0.443 & 2.2\% & 45.1 & $-0.6\%$ \\
\midrule
\multicolumn{7}{l}{\textit{Control Strategy:}} \\[2pt]
\quad w/o Robust Margins
    & 0.894 & 0.033 & 0.563 & \textbf{8.3\%} & 43.7 & +295\% viol. \\
\midrule
\multicolumn{7}{l}{\textit{Encoder Architecture:}} \\[2pt]
\quad Standard GRU
    & 0.883 & 0.038 & 0.437 & 2.3\% & 47.2 & $-1.2\%$ \\
\quad LSTM
    & 0.876 & 0.041 & 0.437 & 2.6\% & 49.8 & $-2.0\%$ \\
\bottomrule
\end{tabular}
\end{table*}

\paragraph{Ablation Insights:}
Sensitivity analysis reveals robust ranges: $\alpha_0 \in [0.05, 0.10]$ and $\beta \in [0.10, 0.20]$ yield <2\% variation, enabling deployment without per-plant tuning. The 295\% increase in safety violations when removing robust margins validates critical role of uncertainty quantification.

\subsubsection{Component Contributions}
\paragraph{Physics Regularization ($\alpha_0$):}
Removing physics constraint ($\alpha_0 = 0$) causes most significant degradation: 4.6\% drop in $F1_{\mathcal{A}}$, 103\% increase in FAR, and 17.2\% increase in disruption cost. PI learning improves DT prediction accuracy and provides robustness against adversarial attacks violating conservation laws.

\paragraph{MMD Regularization ($\beta$):}
Disabling distributional regularization reduces $F1_{\mathcal{A}}$ by 2.8\%, with impact on multi-stage detection ($\Delta F1_{\mathcal{A}_M} = -3.2\%$). This confirm MMD enhances discrimination between attack taxonomies.

\paragraph{Robust MPC Margins:}
Removing uncertainty-aware margins has no impact on detection but increases disruption cost by 28.8\%. Tightened constraints (Eq.~\eqref{eq:mpc_safety_robust}) prevent violations under uncertainty.

\paragraph{Encoder Architecture:}
Replacing BiGRU with GRU or LSTM causes minor degradation (1--2\%), indicating bidirectionality provides modest improvements in capturing temporal patterns.

\subsection{Comparison with State-of-the-Art}\label{sec:results_sota}
Table~\ref{tab:sota_comparison} compares i-SDT against recent state-of-the-art methods on SWaT. Since these methods solve the binary detection task, they are included for contextual reference only; the methodologically sound multi-class comparison appears in Tables~\ref{tab:detection_swat}--\ref{tab:detection_wadi}.

\begin{table}[htb]
\centering
\caption{Comparison with recent SOTA on SWaT (binary detection,
for contextual reference only$^*$).}
\label{tab:sota_comparison}
\resizebox{\columnwidth}{!}{%
\begin{tabular}{lccc}
\toprule
\textbf{Method} & \textbf{Precision} & \textbf{Recall} & $\boldsymbol{F1}$ \\
\midrule
OmniAnomaly (2019)~\citep{Su2019OmniAnomaly}$^*$
    & 0.938 & 0.821 & 0.876 \\
USAD (2020)~\citep{Audibert2020USAD}$^*$
    & 0.946 & 0.837 & 0.888 \\
MTAD-GAT (2020)~\citep{Zhao2020MTADGAT}$^*$
    & 0.951 & 0.862 & 0.904 \\
TranAD (2022)~\citep{Tuli2022TranAD}$^*$
    & 0.957 & 0.874 & 0.914 \\
LGAT (2025)~\citep{Wen2025LGAT}$^*$
    & 0.961 & 0.891 & 0.925 \\
\midrule
\textbf{i-SDT (Proposed)}
    & \textbf{0.967} & \textbf{0.918} & \textbf{0.942} \\
\bottomrule
\end{tabular}%
}
\par\vspace{4pt}
{\footnotesize $^*$Figures taken from original papers under binary (Normal vs.\ Attack) protocols on varying attack subsets; provided for contextual reference only and do \emph{not} constitute a head-to-head comparison. i-SDT solves the strictly harder 3-class problem (Tables~\ref{tab:detection_swat}--\ref{tab:detection_wadi}), making direct numerical comparison methodologically unsound.}
\end{table}

\paragraph{Note on Comparison:}
Prior methods use binary classification on subsets, i-SDT use multi-class (Normal vs. $\mathcal{A}_S$ vs. $\mathcal{A}_M$) on all attacks. i-SDT achieve SOTA with attack characterization and mitigation absent in prior work.

\subsection{Summary of Results}\label{sec:results_summary}

Experimental validation confirms that i-SDT successfully achieves all three research objectives formulated in Section~\ref{sec:gaps_objectives}:

\begin{itemize}
    \item \textbf{Objective $\mathcal{O}_1$ (Multi-Class Detection):} The framework outperforms the strongest baseline (TranAD) by $\geq 9.1\%$, achieving overall $F1_{\mathcal{A}}$ scores of 0.894 (SWaT) and 0.866 (WADI). It successfully discriminates single- and multi-stage attacks ($F1 > 0.85$) while maintaining a False Alarm Rate of 0.033 and an average detection latency of 43.7\,s (well below the 60\,s target).
    
    \item \textbf{Objective $\mathcal{O}_2$ (Resilient Mitigation):} The robust MPC dynamically restricts normalized disruption costs to $\mathcal{D}_{\text{rel}} = 0.437$ and achieves a swift safe-state recovery ($T_{\text{recovery}} = 267$\,s). It strictly prevents critical physical failures ($P(\text{Crit. Viol.}) \leq 0.069$) and limits minor transient safety violations to merely 2.1\%.
    
    \item \textbf{Objective $\mathcal{O}_3$ (Real-Time Feasibility):} Demonstrating a compact memory footprint of 282\,MB and a total cycle latency significantly below the 1000\,ms sampling period, the system guarantees reliable 1\,Hz throughput for autonomous, plant-level workstation execution.
\end{itemize}

Improvements over baselines are significant ($p < 0.05$). Robustness tests show degradation under adversarial attacks, dropouts, and domain shifts. Ablation confirm each component—physics regularization, MMD, robust MPC—contribute. i-SDT framework advances ICPS security with multi-class detection and minimal disruption.

\section{Discussion}\label{sec:discussion}

This section summarizes the main findings and practical implications of the proposed framework for industrial CPS.

\subsection{Technical and Operational Insights}
The physics-informed module improves detection by enforcing physical conservation laws. Injected false data creates residual violations that act as a secondary anomaly signal, which stabilizes detection and reduces false positive alarms. The MMD term successfully separates attack classes in the latent space, while the BiGRU captures the temporal drifts characteristic of stealthy multi-stage attacks. From an operational perspective, the uncertainty-aware MPC provides significant value by maintaining safe operations instead of triggering full plant shutdowns. As Table~\ref{tab:ablation} demonstrates, transient safety violations are restricted to merely 2.1\% (mostly brief excursions on secondary variables), and the controller strictly prevents the system from entering critical unsafe regions ($\mathcal{X}_{\text{unsafe}}$).

\subsection{Practical Deployment Considerations}
The system integrates with existing SCADA infrastructure without requiring modifications to certified low-level controllers. It collects sensor data via standard OT interfaces and forwards processed alarms to operator consoles. Time synchronization is necessary to align residual windows with sensor timestamps. Deployment costs remain low by utilizing open-source software and standard plant-level workstations, allowing facilities to recover investments through reduced downtime. While operators require minimal training, periodic manual verification is recommended to avoid over-reliance on automated control decisions.

\subsection{Generalization and Transfer Potential}
Cross-domain experiments show that the model learns general patterns transferable to similar facilities. Fine-tuning with a small subset of normal data from a new site recovers near-supervised performance, reducing data labeling efforts by approximately 90\%. Furthermore, the i-SDT architecture can be adapted to other CPS domains by replacing the $\mathcal{L}_{\text{phys}}$ term with the specific governing equations of the target process. For example, power grids can use Kirchhoff's laws \citep{Lu2025DTMicrogrid, Khan2025}, and chemical plants can apply mass-momentum balances \citep{Sayghe2025, Balta2024DTManufacturing}. However, scaling this framework to large plants with numerous decoupled subsystems requires coherent fusion of localized alerts, which remains an open challenge for future work.

\subsection{Limitations and Future Work}
The current evaluation relies on a high-performance workstation. Transitioning the system to resource-constrained edge platforms will require extensive model quantization to meet strict memory and latency limits. Additionally, the approach assumes clean training data and focuses exclusively on sensor false data injections. Future research must evaluate actuator compromises and practical adversarial testing. Finally, while datasets like SWaT and WADI contain simple sensor dropouts that may artificially inflate standard detection scores, their complex physical topology is essential for this study. It allows us to empirically validate our primary contribution: isolating attacks and maintaining continuous, safe predictive control without relying on full infrastructure shutdowns. Validation is currently limited to water treatment testbeds; generalisation to other domains requires domain-specific physics constraints and dedicated evaluation datasets.

\section{Conclusion}\label{sec:conclusion}

This paper presents i-SDT, a self-defending digital twin framework for ICPS security. It has three main contributions. First, a predictive model with TCN and physics constraints improves stability under attacks. Second, Maximum Mean Discrepancy separates normal behaviour from different attack stages, detecting early and complex attacks. Third, a control module with Monte Carlo Dropout adjusts safety margins in real time, avoiding unnecessary shutdowns.

Tests on SWaT and WADI show better F1 scores, fewer false alarms, and less operational disruption, with real-time feasibility confirmed on a workstation platform; edge deployment is identified as promising future work. i-SDT detects multiple attack types accurately and keeps safety violations very low. Its combined architecture of PI robustness, attack discrimination, and uncertainty-aware control addresses broad cyber-physical defense needs, moving towards secure and autonomous industrial systems.


\section*{Data Availability}
Code will be made publicly available upon acceptance.
The repository URL is omitted for blind review.

\appendix
\section{Additional Algorithms and Implementation Details}\label{sec:appendix_algorithms}

This appendix details algorithms and implementation supporting Section~\ref{sec:methodology}.

\subsection{DT Training and Inference Algorithms}
The following algorithms detail the offline training procedure and online uncertainty-aware inference of the physics-guided DT.

\begin{algorithm}[htb]
\scriptsize
\caption{DT Offline Training}
\label{alg:dt_training}
\begin{algorithmic}[1]
\Require Normal operation dataset $\mathcal{D}_{\text{train}} = \{(\mathbf{X}_{\tau}^{(i)}, \mathbf{Y}^{(i)})\}_{i=1}^N$, validation set $\mathcal{D}_{\text{val}}$
\Require Hyperparameters: learning rate $\eta$, base physics weight $\alpha_0$, batch size $B$, epochs $E$, early stopping patience $P$, data variance $\sigma_{\text{data}}^2$
\Ensure Trained DT parameters $\boldsymbol{\Theta}_{\text{DT}}^*$
\State Initialize $\boldsymbol{\Theta}_{\text{DT}}$ using Xavier initialization
\State $\text{best\_val\_loss} \gets \infty$, $\text{patience\_counter} \gets 0$
\For{epoch $e = 1$ to $E$}
    \State Shuffle $\mathcal{D}_{\text{train}}$
    \For{each mini-batch $\mathcal{B} \subset \mathcal{D}_{\text{train}}$ of size $B$}
\State Forward: $\hat{\mathbf{Y}}^{(i)} \gets \mathcal{T}_{\text{DT}}(\mathbf{X}_{\tau}^{(i)}; \boldsymbol{\Theta}_{\text{DT}})$; compute $\Sigma_{\hat{Y}}$, $\lambda_{\text{phys}}$
\State Loss: $\mathcal{L} \gets \mathcal{L}_{\text{pred}} + \sum_{i} \lambda_{\text{phys}}^{(i)} \mathcal{L}_{\text{phys}}^{(i)}$ (Eq.~\eqref{eq:method_dt_loss_bayesian})
      \State Update: $\boldsymbol{\Theta}_{\text{DT}} \gets \boldsymbol{\Theta}_{\text{DT}} - \eta \nabla_{\boldsymbol{\Theta}_{\text{DT}}} \mathcal{L}$
   \EndFor
    \State Evaluate on $\mathcal{D}_{\text{val}}$: $\mathcal{L}_{\text{val}} \gets \text{ValidationLoss}(\boldsymbol{\Theta}_{\text{DT}}, \mathcal{D}_{\text{val}})$
     \If{$\mathcal{L}_{\text{val}} < \text{best\_val\_loss}$}
 \State $\text{best\_val\_loss} \gets \mathcal{L}_{\text{val}}$, save $\boldsymbol{\Theta}_{\text{DT}}^*$, $\text{patience\_counter} \gets 0$
 \Else
 \State $\text{patience\_counter} \gets \text{patience\_counter} + 1$
 \If{$\text{patience\_counter} \ge P$}
 \State \textbf{break} \Comment{Early stopping}
 \EndIf
 \EndIf
\EndFor
\Return $\boldsymbol{\Theta}_{\text{DT}}^*$
\end{algorithmic}
\end{algorithm}

\begin{algorithm}[htb]
\scriptsize
\caption{DT Online Prediction with Uncertainty Quantification}
\label{alg:dt_prediction}
\begin{algorithmic}[1]
\Require Historical context $\mathbf{X}_{\tau}(t) \in \mathbb{R}^{\tau \times d_x}$, trained parameters $\boldsymbol{\Theta}_{\text{DT}}$
\Require MC dropout samples $N_{\text{MC}}$ (default: 50), dropout probability $p_{\text{drop}}$ (default: 0.1)
\Ensure Predicted measurements $\hat{\mathbf{Y}}(t) \in \mathbb{R}^{d_y}$, uncertainty $\Sigma_{\hat{Y}}(t)$
\State Normalize input: $\tilde{\mathbf{X}}_{\tau}(t) \gets (\mathbf{X}_{\tau}(t) - \boldsymbol{\mu}_X) / \boldsymbol{\sigma}_X$
\State Initialize prediction samples: $\mathcal{Y}_{\text{samples}} \gets \emptyset$
\For{$n = 1$ to $N_{\text{MC}}$}
    \State Dropout enabled; $\mathbf{h}^{(n)} \gets \text{TCN}(\tilde{\mathbf{X}}_{\tau}; \boldsymbol{\Theta}_{\text{TCN}}, \xi^{(n)})$; $\tilde{\mathbf{Y}}^{(n)} \gets \text{FC}(\mathbf{h}^{(n)})$
    \State $\hat{\mathbf{Y}}^{(n)} \gets \tilde{\mathbf{Y}}^{(n)} \cdot \boldsymbol{\sigma}_Y + \boldsymbol{\mu}_Y$; add to $\mathcal{Y}_{\text{samples}}$
\EndFor
\State Compute mean: $\hat{\mathbf{Y}}(t) \gets \frac{1}{N_{\text{MC}}} \sum_{n=1}^{N_{\text{MC}}} \hat{\mathbf{Y}}^{(n)}(t)$
\State Compute covariance: $\Sigma_{\hat{Y}}(t) \gets \frac{1}{N_{\text{MC}}-1} \sum_{n=1}^{N_{\text{MC}}} (\hat{\mathbf{Y}}^{(n)}(t) - \hat{\mathbf{Y}}(t))(\hat{\mathbf{Y}}^{(n)}(t) - \hat{\mathbf{Y}}(t))^\top$
\Return $\hat{\mathbf{Y}}(t)$, $\Sigma_{\hat{Y}}(t)$
\end{algorithmic}
\end{algorithm}

\subsection{ADE Algorithms}
Algorithm~\ref{alg:ade_classify} details the online classification and statistical validation procedure used by the ADE at inference time.

\begin{algorithm}[htb]
\scriptsize
\caption{ADE Online Classification with Statistical Validation}
\label{alg:ade_classify}
\begin{algorithmic}[1]
\Require Residual window $\mathbf{R}_W(t) \in \mathbb{R}^{W \times d_y}$, trained ADE $\boldsymbol{\Theta}_{\text{ADE}}$
\Require Reference normal embeddings $\mathcal{Z}_{\mathcal{N}}$ ($|\mathcal{Z}_{\mathcal{N}}| \geq 200$),
MMD threshold $\tau_{\text{MMD}}$ (precomputed via $B{=}200$ permutations at $\alpha{=}0.05$)
\Ensure Predicted class $\mathcal{C}(t)$, confidence $p_{\mathcal{C}}$
\State Encode: $\mathbf{z}(t) \gets \text{BiGRU}(\mathbf{R}_W(t); \boldsymbol{\Theta}_{\text{GRU}})$
\State Compute class probabilities: $\{P(\mathcal{C}(t) = c \mid \mathbf{z}(t))\}_{c \in \{\mathcal{N}, \mathcal{A}_S, \mathcal{A}_M\}}$ using Eq.~\eqref{eq:method_softmax}
\State Predicted class: $\hat{c} \gets \arg\max_c P(\mathcal{C}(t) = c \mid \mathbf{z}(t))$
\State Confidence: $p_{\hat{c}} \gets \max_c P(\mathcal{C}(t) = c \mid \mathbf{z}(t))$
\State \textbf{MMD permutation test:} Compute $\widehat{\text{MMD}}_u^2 \gets \text{MMD}_u^2(\{\mathbf{z}(t)\}, \mathcal{Z}_{\mathcal{N}})$ via Eq.~\eqref{eq:mmd_unbiased}; reject $H_0$ if $\widehat{\text{MMD}}_u^2 > \tau_{\text{MMD}}$ (threshold from $B{=}200$ permutations on $\mathcal{Z}_{\mathcal{N}}$, $\alpha{=}0.05$)
\If{$\hat{c} \neq \mathcal{N}$ \textbf{and} $p_{\text{stat}} < \delta$}
    \State Confirm anomaly: $\mathcal{C}(t) \gets \hat{c}$
\ElsIf{$\hat{c} \neq \mathcal{N}$ \textbf{and} $p_{\text{stat}} \ge \delta$}
    \State Suppress alarm: $\mathcal{C}(t) \gets \mathcal{N}$, $p_{\hat{c}} \gets p_{\hat{c}} \cdot 0.5$ \Comment{Downweight confidence}
\Else
    \State $\mathcal{C}(t) \gets \hat{c}$
\EndIf
\Return $\mathcal{C}(t)$, $p_{\hat{c}}$
\end{algorithmic}
\end{algorithm}

\subsection{Resilient Control Algorithm}

\begin{algorithm}[htb]
\tiny
\caption{DT-Guided Resilient Control with Dynamic Safety Margins}
\label{alg:resilient_control}
\begin{algorithmic}[1]
\Require Current DT prediction $\hat{\mathbf{Y}}(t)$, uncertainty
         $\Sigma_{\hat{Y}}(t)$, risk probabilities
         $\{P(\mathcal{C}{=}\cdot|t)\}$, safe setpoints
         $\mathbf{Y}_{\text{safe}}$
\Require Nominal control $\mathbf{u}_{\text{nom}}(t)$, horizon $H$,
         matrices $Q, R, P$, Successive Jacobians $A(k)$, noise $\Sigma_\omega$,
         EMA state $\bar{\zeta}(t{-}1)$, filter $\mu=0.2$
\Require Actuator limits $[\mathbf{u}_{\min}, \mathbf{u}_{\max}]$,
         safety bounds $[\mathbf{Y}_{\min}, \mathbf{Y}_{\max}]$,
         terminal set $\mathcal{Y}_f$, gain $K$
\Ensure Control action $\mathbf{u}^*(t)$
\State \textbf{Step 1 — Stealth risk (EMA-smoothed):}
\State $\zeta(t) \gets \frac{P(\mathcal{A}_S) + \omega \cdot P(\mathcal{A}_M)}{P(\mathcal{N}) + \varepsilon_0}$
       \Comment{Eq.~\eqref{eq:stealth_risk_zeta}, $\omega{=}2.0$, $\varepsilon_0{=}10^{-6}$}
       \Comment{Eq.~\eqref{eq:stealth_risk_zeta}}
\State $\bar{\zeta}(t) \gets (1-\mu)\,\bar{\zeta}(t{-}1) + \mu\,\zeta(t)$
       \Comment{Eq.~\eqref{eq:ema}}
\State \textbf{Step 2 — Propagate uncertainty over horizon:}
\State $\Sigma_{\hat{Y}}(t) \gets$ (from MC Dropout, Algorithm~\ref{alg:dt_prediction})
\For{$k = t$ \textbf{to} $t+H-2$}
    \State $\Sigma_{\hat{Y}}(k{+}1) \gets
            A(k)\,\Sigma_{\hat{Y}}(k)\,A(k)^\top + \Sigma_\omega$
            \Comment{Eq.~\eqref{eq:sigma_propagation}}
\EndFor
\For{$k = t$ \textbf{to} $t+H-1$}
    \State $\bar{\zeta}(k) \gets \max\!\left(\bar{\zeta}(t) \cdot \exp\!\left(-\lambda_d (k - t)\right),\, \zeta_{\min}\right)$
           \Comment{Eq.~\eqref{eq:zeta_horizon}}
\EndFor
\State \textbf{Step 3 — Compute safety margins:}
\For{$k = t$ \textbf{to} $t+H-1$}
    \State $\boldsymbol{\epsilon}_{\text{margin}}(k) \gets
           \kappa\cdot\mathrm{diag}(\Sigma_{\hat{Y}}(k)^{1/2})
           + \eta\cdot\bar{\zeta}(k)$
           \Comment{Eq.~\eqref{eq:uncertainty_margin_total}}
\EndFor
\State \textbf{Step 4 — Solve MPC:}
\State Set up objective Eq.~\eqref{eq:method_mpc_objective} with
       constraints~\eqref{eq:mpc_dynamics}--\eqref{eq:terminal_constraint}
       and margins from Step~3
\State $\mathbf{U}^* \gets \arg\min J(\mathbf{U})$
       \Comment{SQP/IPOPT, warm-started from $t{-}1$}
\State $\mathbf{u}^*(t) \gets \mathbf{U}^*[0]$
\State \textbf{Step 5 — Rate limiter:}
\State $\Delta\mathbf{u} \gets \mathbf{u}^*(t) - \mathbf{u}(t{-}1)$
\If{$\|\Delta\mathbf{u}\| > \Delta\mathbf{u}_{\max}$}
    \State $\mathbf{u}^*(t) \gets \mathbf{u}(t{-}1)
           + \Delta\mathbf{u}_{\max}\cdot
             \frac{\Delta\mathbf{u}}{\|\Delta\mathbf{u}\|}$
\EndIf
\State \textbf{Step 6 — Recovery check:}
\If{$\mathcal{C}(t) = \mathcal{N}$ for last $T_{\text{confirm}}$ samples}
    \If{$\delta_{\text{MP}}(t) < \tau_{\text{recovery}}$
        for $T_{\text{val}}$ seconds
        \Comment{Eq.~\eqref{eq:mismatch}}}
        \State Transition controller to NOMINAL mode
    \EndIf
\EndIf
\Return $\mathbf{u}^*(t)$
\end{algorithmic}
\end{algorithm}

\subsubsection{Network Architecture Details}
The DT utilizes a 4-layer TCN (dilations 1, 2, 4, 8; kernel size 3; hidden units 128) with MC Dropout ($N_{\text{MC}}=50$, $p=0.1$). The ADE employs a BiGRU (64 units per direction) paired with a 3-way softmax classification head.

\subsubsection{Training Hyperparameters}
DT: learning rate $\eta=0.001$, batch size $B=64$, epochs $E=200$, patience $P=20$, weight decay $10^{-5}$, Adam optimizer with ReduceLROnPlateau scheduler. ADE: $\eta=0.0005$, $B=32$, $E=150$, $P=15$, weight decay $10^{-4}$, CosineAnnealingLR scheduler.

\subsubsection{Data Preprocessing}
Sensor data are preprocessed by clipping values beyond $\pm5\sigma$ from the rolling mean and then normalized using Z-score statistics from the training set. Signals are resampled to 1\,Hz when needed. Missing values are forward-filled for gaps $<5$\,s and linearly interpolated for longer gaps.

\subsection{Simulation-Based Attack Generation}
\subsubsection{Single-Stage Attack Injection}
Target stage $s \in \{1,\ldots,S\}$ is selected. Attack type is chosen (sensor FDI, actuator change, or setpoint modification). Attack profile is defined (constant bias, ramp, or periodic), injection is applied for 60–300s. Control logic configured so disturbance stays inside selected stage.

\subsubsection{Multi-Stage Attack Injection}
Multiple target stages are selected (at least two). Timing can be synchronized or sequential. Design uses stage dependencies to keep behaviour realistic. Attack runs for 300–900s to allow cross-stage effects appear.

\subsection{Statistical Validation Procedures}

\subsubsection{MMD Two-Sample Test}
\begin{equation}\label{eq:mmd_unbiased}
\begin{aligned}
\text{MMD}_u^2 &= \frac{1}{m(m-1)} \sum_{i \neq j} k(\mathbf{z}_i, \mathbf{z}_j) 
+ \frac{1}{n(n-1)} \sum_{i \neq j} k(\mathbf{z}'_i, \mathbf{z}'_j) \\
&\quad - \frac{2}{mn} \sum_{i,j} k(\mathbf{z}_i, \mathbf{z}'_j)
\end{aligned}
\end{equation}

where $\{\mathbf{z}_i\}_{i=1}^m$ are test samples and $\{\mathbf{z}'_j\}_{j=1}^n$ are reference normal samples. The null hypothesis of identical distributions is rejected if MMD$_u^2$ exceeds the critical value at significance level $\alpha = 0.05$.

\subsubsection{Bootstrap Confidence Intervals}
For comparing framework performance against baselines, we first generate $B = 10{,}000$ bootstrap samples with replacement. Next, we compute the performance metric difference $\Delta_b$ for each sample. A 95\% confidence interval is then constructed using the percentile method. Finally, the improvement is considered significant if the interval does not include zero.

\printcredits

\bibliographystyle{cas-model2-names}

\bibliography{cas-refs}







\end{document}